\documentclass{article}
\usepackage{graphicx} 
\usepackage{amsmath}
\usepackage{amsfonts}
\usepackage{amsthm}
\usepackage{comment}
\usepackage[utf8]{inputenc}
\usepackage[toc]{appendix}
\usepackage{algorithm}
\usepackage{algpseudocode}
\usepackage{booktabs}
\usepackage{color}
\usepackage{diagbox}
\usepackage{amssymb}
\usepackage{threeparttable}
\usepackage{makecell}
\usepackage{enumerate}
\usepackage{enumitem}
\usepackage{float}
\usepackage{multirow}
\usepackage{fullpage}
\usepackage{makecell}
\usepackage{appendix}
\usepackage{longtable}
\usepackage{authblk}

\usepackage[
backend=biber,
style=alphabetic,
sorting=nty
]{biblatex}

\newcommand{\nl}{\mathrm{nl}}

\newcommand{\wt}{\mathrm{wt}}

\newcommand{\tr}{\mathrm{tr}}
\newcommand{\solnum}{\mathrm{N}}
\newcommand{\solnumone}{\mathrm{N}_1}
\newcommand{\solnumtwo}{\mathrm{N}_2}

\newcommand{\numroots}{\mathrm{N}}

\newtheorem{lemma}{Lemma}
\newtheorem{theorem}{Theorem}
\newtheorem{claim}{Claim}
\newtheorem{corollary}{Corollary}
\newtheorem{definition}{Definition}

\newtheorem{proposition}{Proposition} 
\newtheorem{remark}{Remark}

\title{Trace Monomial Boolean Functions with Large High-Order Nonlinearities}

\author{
Jinjie Gao\footnote{School of Computer Science, Fudan University, Shanghai 200433, China. Email: jjgao18@fudan.edu.cn},
Haibin Kan\footnote{Shanghai Key Laboratory of Intelligent Information Processing, School of Computer Science, Fudan University, Shanghai 200433, China;
Shanghai Engineering Research Center of Blockchain, Shanghai 200433, China;
Yiwu Research Institute of Fudan University, Yiwu City 322000, China. Email: hbkan@fudan.edu.cn},
Yuan Li\footnote{School of Computer Science, Fudan University, Shanghai 200433, China. Email: yuan\_li@fudan.edu.cn},
Jiahua Xu \footnote{School of Computer Science, Fudan University, Shanghai 200433, China. Email: jiahuaxu21@m.fudan.edu.cn},
Qichun Wang\footnote{School of Computer Science and Technology, Nanjing Normal University, China. Email: qcwang@fudan.edu.cn}}

\begin{document}
\date{}
\maketitle

\begin{abstract}
Exhibiting an explicit Boolean function with a large high-order nonlinearity is an important problem in cryptography, coding theory, and computational complexity. We prove lower bounds on the second-order, third-order, and higher-order nonlinearities of some trace monomial Boolean functions.

We prove lower bounds on the second-order nonlinearities of functions $\mathrm{tr}_n(x^7)$ and $\mathrm{tr}_n(x^{2^r+3})$ where $n=2r$. Among all trace monomials, our bounds match the best second-order nonlinearity lower bounds by \cite{Car08} and \cite{YT20} for odd and even $n$ respectively. We prove a lower bound on the third-order nonlinearity for functions $\mathrm{tr}_n(x^{15})$, which is the best third-order nonlinearity lower bound. For any $r$, we prove that the $r$-th order nonlinearity of  $\mathrm{tr}_n(x^{2^{r+1}-1})$ is at least $2^{n-1}-2^{(1-2^{-r})n+\frac{r}{2^{r-1}}-1}- O(2^{\frac{n}{2}})$. For $r \ll \log_2 n$, this is the best lower bound among all explicit functions.
\end{abstract}


\section{Introduction}

Exhibiting an \emph{explicit} Boolean function with a large \emph{high-order nonlinearity} is an important task in areas including cryptography, coding theory, and computational complexity. In cryptography, a high nonlinearity is an important cryptographic criterion for Boolean functions used in symmetric-key cryptosystems to resist correlation attacks \cite{Car21}. In coding theory, the largest $r$-th order nonlinearity among all $n$-variable Boolean functions is exactly the covering radius of Reed-Muller codes $\mathrm{RM}(r, n)$; computing (high-order) nonlinearity is related to the problem of decoding Reed-Muller codes. In computational complexity, one \emph{must} prove large enough nonlinearity lower bound (for a function in NP) to prove that NP does not have circuits of quasi-polynomial size \cite{GHR92, Vio22}. In addition, this problem is related to pseudorandom generators, communication complexity, and circuit complexity; we send interested readers to the survey by Viola  \cite{Vio22}.


Known techniques for proving nonlinearity lower bound include Hilbert function \cite{Raz87, Smo87}, the ``squaring trick'' \cite{BNS92, Gow98, Gow01, Car08}, XOR lemmas \cite{Bou05, GRS05, Vio06, VW08, CHHLZ20, Chen21}, invariant theory \cite{DV22}, symmetrization \cite{IPV23}, etc. In this work, we follow the ``squaring trick'' methods by Carlet \cite{Car08} to prove nonlinearity lower bounds for \emph{trace monomial} functions. Trace monomials are good candidates to study, both experimentally and theoretically.

Carlet \cite{Car08} proposed a method to lower bound the $r$-th order nonlinearity by estimating the minimum $(r-1)$-th order nonlinearities for all its derivatives; applying this for $r-1$ times, the $r$-th order nonlinearity can be lower bounded by the minimum (first-order) nonlinearity for all its $(r-1)$-th order derivatives.
Canteaut \emph{et al.} \cite{CCK08} provided a method to determine the Walsh spectrum (and thus the nonlinearity) of any \emph{quadratic} function by the dimension of its \emph{linear kernel}. In this way, the problem of lowering bound nonlinearity essentially reduces to the problem of estimating the number of roots of certain equations over finite fields.
Along this line, nonlinearity lower bounds for trace monomial Boolean functions are proved in \cite{GG09, SG09, SW09, GG10, GST10, Car11, GG11b, LHG11, SW11,TCT13, Sin14, GT18, MKJ20, TYZZ20,  Liu21,  SG22, SG23, TS23}. We summarize the second-order nonlinearity lower bounds in Table \ref{funcs_second}.

\begin{center}
\centering
\label{funcs_second}
 \renewcommand\arraystretch{2}  
 \begin{longtable}{c|c}
  \caption{Second-order nonlinearity lower bounds}\\
\hline
 \textbf{Function} & \textbf{$\nl_2$ lower bound}  \\ \hline

 \thead{$\tr_n(\mu x^{2^t-1}+g(x))$, $\mu \in \mathbb{F}_{2^n}^*$,\\ $t\le n$ \cite{Car08}$^{\rm a}$} & 
$
\begin{array}{l l}{{}}&{{2^{n-1}-\frac{1}{2}\sqrt{(2^{n}-1)(2^t-4)2^{\frac{n}{2}}+2^n}}}\\ {{\ge}}&{{2^{n-1}-2^{\frac{3n}{4}+\frac{t}{2}-1}-O(2^{\frac{n}{4}})}}\end{array} 
$
   \\ \hline
 
 $\tr_n(x^{2^r+3}),n=2r+1$ \cite{Car08} & 
$
\begin{array}{l l}{{}}&{{2^{n-1}-\frac{1}{2}\sqrt{(2^{n}-1)2^{\frac{n+5}{2}}+2^{n}}}}\\ {{=}}&{{2^{n-1}-2^{\frac{3n+1}{4}}-O(2^{\frac{n}{4}})}}\end{array} 
$
   \\ \hline
 
 $\tr_n(x^{2^r+3}),n=2r-1$ \cite{Car08}& 
    
$ 
\begin{aligned}
&
\begin{cases}
2^{n-1}-{\frac{1}{2}}\sqrt{2^{\frac{3n+1}{2}}+2^{\frac{3n-1}{2}}+2^{n}-2^{\frac{n+3}{2}}}, &{\text{if}\ 3\nmid n} 
\\ 2^{n-1}-\frac{1}{2}\sqrt{3\cdot 2^{\frac{3n-1}{2}}+2^n+3\cdot 2^{n+\frac{1}{2}}-2^{\frac{n+3}{2}}}, &{\text{if}\ 3\mid n}
\end{cases}\\
= & \quad 2^{n-1}-2^{\frac{3n-5}{4}+\frac{1}{2}\log_2 3}-O(2^{\frac{n}{4}})
\end{aligned}
$ 

 \\ \hline

 $\tr_n(x^{2^n-2})$ \cite{Car08}& 
$
\begin{array}{l l}{{}}&{2^{n-1}-\frac{1}{2}\sqrt{(2^{n}-1)2^{\frac{n}{2}+2}+3\cdot2^{n}}}\\ {{=}}&{{2^{n-1}-2^{\frac{3n}{4}}-O(2^{\frac{n}{4}})}}\end{array} 
$
 \\ \hline

 \thead{$\tr_{\frac{n}{2}}(xy^{2^n-2})$\\ $x,y\in \mathbb{F}_{\frac{n}{2}}$, \text{if n is even} \cite{Car09}}& 
$
\begin{array}{l l}{{}}&{2^{n-1}-\frac{1}{2}\sqrt{2^n+(2^{n+2}+2^{\frac{3n}{4}+1}+2^{\frac{n}{2}+1})(2^{\frac{n}{2}}-1)}}\\ {{=}}&{{2^{n-1}-2^{\frac{3n}{4}}-O(2^{\frac{n}{2}})}}\end{array} 
$
 \\ \hline

 $\tr_n(\mu x^{2^{i}+2^j+1}), \mu \in \mathbb{F}_{2^n}^*$ \cite{GG09}& 
 $ \left\{ \begin{array}{l l}{{2^{n-1}-2^{\frac{3n+2i-4}{4}}}},&\text{if n is even} \\  {{2^{n-1}-2^{\frac{3n+2i-5}{4}}}},&\text{if n is odd} \end{array}\right.$   \\ \hline

 \thead{$\tr_n(\mu x^{2^{2i}+2^i+1}), \mu \in \mathbb{F}_{2^n}^*$,\\ $\gcd(n,i)=1$, $n>4$ \cite{GG09}}&
$\left\{\begin{array}{l l}{{2^{n-1}-2^{\frac{3n}{4}}}},&\text{if n is even}\\ {{2^{n-1}-2^{\frac{3n-1}{4}}}},&\text{if n is odd}\end{array}\right.$ 
  \\ \hline

 \thead{$\tr_n(\lambda x^{2^{2r}+2^r+1})$, $n=6r$, \\ $\lambda \in \mathbb{F}_{2^n}^*$ \cite{GST10} }&
$\begin{array}{l l}{{}}&{
2^{n-1}-\frac{1}{2}\sqrt{2^{\frac{3n}{2}+2r}+2^n-2^{\frac{n}{2}+2r}}}
\\{{=}}&{2^{n-1}-2^{\frac{3n}{4}+r-1}-O(2^{\frac{n}{4}})}\end{array}
$ 
 \\ \hline

  \thead{$\tr_r(xy^{2^i+1})$\\ $n=2r$, $x,y \in \mathbb{F}_{2^{r}}$, $1\le i<r$, \\ $\gcd(2^r-1,2^i+1)=1$, $\gcd(i,r)=j$ \cite{GST10} }&
$\begin{array}{l l}{{}}&{
2^{n-1}-\frac{1}{2}\sqrt{2^{\frac{3n}{2}+j}-2^{\frac{3n}{4}+\frac{j}{2}}+2^n(2^{\frac{n}{4}+\frac{j}{2}}-2^j+1)}}
\\{{=}}&{2^{n-1}-2^{\frac{3n}{4}+\frac{j}{2}-1}-O(2^{\frac{n}{2}})}\end{array}
$ 
 \\ \hline
   \thead{$\tr_n(\lambda x^{2^{2r}+2^r+1})$, $n=5r$, \\  $\lambda \in \mathbb{F}_{2^r}^*$ \cite{GG11b} }&
$2^{n-1}-2^{\frac{3n+3r-4}{4}}$ 
\\ \hline

 \thead{$\tr_n(\lambda x^{2^{2r}+2^r+1})$, $n=3r$, \\  $\lambda \in \mathbb{F}_{2^r}^*$ \cite{Sin11} }&
$2^{n-1}-2^{\frac{3n+r-4}{4}} $ 
 \\ \hline

 \thead{$\tr_n(\lambda x^{2^{2r}+2^r+1})$, $n=4r$,\\ $\lambda \in \mathbb{F}_{2^r}^*$ \cite{SW11} }&
$\begin{array}{l l}{{}}&{
2^{n-1}-\frac{1}{2}\sqrt{2^{\frac{7n}{4}}+2^{\frac{5n}{4}}-2^n}}
\\{{=}}&{2^{n-1}-2^{\frac{7n}{8}-1}-O(2^{\frac{3n}{8}})}\end{array} $ 
 \\ \hline
 
 \thead{$\tr_n(\lambda x^{2^{2r}+2^r+1})$, $n=6r$, \\ $\lambda \in \mathbb{F}_{2^n}^*$ \cite{TYZZ20}$^{\rm b}$}&
$\begin{array}{l l}{{}}&{
2^{n-1}-\frac{1}{2}\sqrt{2^{\frac{5n}{3}}+2^{\frac{4n}{3}}-2^{\frac{7n}{6}}+2^{n}-2^{\frac{5n}{6}}} }
\\{{=}}&{2^{n-1}-2^{\frac{5n}{6}-1}-O(2^{\frac{n}{2}})}\end{array}$ 
 \\ \hline

 $\tr_n(x^{2^{r+1}+3}),n=2r$ \cite{YT20}& 
$ 
\begin{aligned}
&
\begin{cases}
    2^{n-1}-\frac{1}{2}\sqrt{2^{\frac{3n}{2}+1}+2^{\frac{5n}{4}+\frac{1}{2}}-2^{n}-2^{\frac{3n}{4}+\frac{1}{2}}} ,&\text{if r is odd}\\ 
    2^{n-1}-\frac{1}{2}\sqrt{2^{\frac{3n}{2}+1}+{\frac{1}{3}}\cdot2^{\frac{5n}{4}+2}-2^{n}-\frac{1}{3}\cdot 2^{\frac{3n}{4}+2}},&\text{if r is even}
\end{cases}
\\= & \quad 2^{n-1}-2^{\frac{3n}{4}-\frac{1}{2}}-O(2^{\frac{n}{2}})
\end{aligned}
$ 
  \\ \hline

 \thead{$\tr_n(x^{2^r+2^{\frac{r+1}{2}}+1}),n=2r$,\\ \text{for odd r} \cite{YT20}}&
$\begin{array}{l l}{{}}&{
2^{n-1}-\frac{1}{2}\sqrt{2^{\frac{3n}{2}+1}+2^{\frac{5n}{4}+\frac{1}{2}}-2^{n}-2^{\frac{3n}{4}+\frac{1}{2}}}}
\\{{=}}&{2^{n-1}-2^{\frac{3n}{4}-\frac{1}{2}}-O(2^{\frac{n}{2}})}\end{array}
$
  \\ \hline

 \thead{$\tr_n(x^{2^{2r}+2^{r+1}+1}),n=4r$, \\ \text{for even r} \cite{YT20}}&
$\begin{array}{l l}{{}}&{
2^{n-1}-\frac{1}{2}\sqrt{2^{\frac{3n}{2}+1}+\frac{1}{3}\cdot 2^{\frac{5n}{4}+2} -2^{n}-\frac{1}{3}\cdot 2^{\frac{3n}{4}+2}}}
\\{{=}}&{2^{n-1}-2^{\frac{3n}{4}-\frac{1}{2}}-O(2^{\frac{n}{2}})}\end{array}
$ 
 \\ \hline


 \thead{$\tr_n(x^{2^{r+1}+2^{r}+1}),n=2r+2$,\\ \text{for even r} \cite{Liu21}}&
$\begin{array}{l l}{{}}&{
2^{n-1}-{\frac{1}{2}}\sqrt{2^{\frac{3n}{2}+1}+2^{\frac{5n}{4}+\frac{1}{2}}-2^{n}-2^{{\frac{3n}{4}}+\frac{1}{2}}} }
\\{{=}}&{2^{n-1}-2^{\frac{3n}{4}-\frac{1}{2}}-O(2^{\frac{n}{2}})}\end{array}
$
 \\ \hline
 
\end{longtable}

\begin{threeparttable}
\begin{tablenotes}
    \footnotesize

    \item [$^{\rm a}$] $g(x)$ is a univariate polynomial of degree $\le 2^t-2$ over $\mathbb{F}_{2^n} $.
    \item [$^{\rm b}$] $\lambda\in \{yz^{d}: y\in U, z\in \mathbb{F}_{2^n}^{*}\}$, $U=\{ y\in \mathbb{F}_{2^{3r}}^{*} : \tr_{\mathbb{F}_{2^{3r}}/\mathbb{F}_{2^{r}}}(y)=0\}$, where the function $\tr_{\mathbb{F}_{2^{3r}}/\mathbb{F}_{2^{r}}}(y)$ is a mapping from $\mathbb{F}_{2^{3r}}$ to $\mathbb{F}_{2^{r}}$.
\end{tablenotes}

\end{threeparttable}
\end{center}

Among all trace monomials, the best second-order nonlinearity lower bound was proved for functions $\tr_n(x^{2^r+3})$, where $n=2r-1$, by Carlet \cite{Car08}, when $n$ is odd, and for functions  $\tr_n(x^{2^{r+1}+3})$, where $n=2r$, by Yan and Tang \cite{YT20}, when $n$ is even.
Note that the best second-order nonlinearity lower bound is, $2^{n-1} - 2^{\frac{2}{3} \dot n} - 2^{\frac{1}{3}n-1}$, proved by Kolokotronis and Limniotis \cite{KL11}, for the Maiorana-McFarland cubic functions (which are \emph{not} trace monomials). 

For the third-order nonlinearity, lower bounds are proved for the inverse function $\mathrm{tr}_n(x^{2^{n}-2})$, the Kasami functions $\tr_n(\mu x^{57})$,  functions of the form $\tr_n(\mu x^{2^i+2^j+2^k+1})$. Previous to our results, the best third-order nonlinearity lower bound was proved for functions $\tr_n(\mu x^{2^{3i}+2^{2i}+2^i+1})$, where  $\mu \in \mathbb{F}_{2^n}^*$, and $\gcd(i,n)=1$, by Singh \cite{Sin14}. Please see Table \ref{funcs_third} for a summary.

\begin{table}[H]
\centering
\caption{Third-order nonlinearity lower bounds}
 \renewcommand\arraystretch{2}
\begin{tabular}{c|c}
\hline
 \textbf{Function} & \textbf{$\nl_3$ lower bound}  \\ \hline

 $\tr_n(x^{2^n-2})$ \cite{Car08}& 
$\begin{array}{l l}{{}}&{2^{n-1}-\frac{1}{2}\sqrt{(2^n-1)\sqrt{2^{\frac{3n}{2}+3}+3\cdot2^{n+1}-2^{\frac{n}{2}+3}+16}+2^n}}\\ {{=}}&{{2^{n-1}-2^{\frac{7n}{8}-\frac{1}{4}}-O(2^{\frac{3n}{8}})}}\end{array} $
 \\ \hline

\thead{ $\tr_n(\mu x^{57}), \mu \in \mathbb{F}_{2^n}^*$\\ $n>10$ \cite{GG10}}& 
 $\left\{\begin{array}{l l}{{2^{n-3}-2^{\frac{n+4}{2}}}},&{\text{if n is even} }\\ {{2^{n-3}-2^{\frac{n+3}{2}}}},&{\text{if n is odd}}\end{array}\right.$   \\ \hline

 \thead{$\tr_n(\mu x^{2^i+2^j+2^k+1})$, $i>j>k\ge 1$, \\$n>2i$, $ \mu \in \mathbb{F}_{2^n}^*$ \cite{Sin14}}&
$\left\{\begin{array}{l l}{{2^{n-3}-2^{\frac{n+2i-6}{2}}},}&{\text{if n is even} }\\ {{2^{n-3}-2^{\frac{n+2i-7}{2}}}},&{\text{if n is odd}}\end{array}\right.$ 
  \\ \hline

\thead{ $\tr_n(\mu x^{2^{3i}+2^{2i}+2^i+1}), \mu \in \mathbb{F}_{2^n}^*$ \\ $\gcd(i,n)=1$, $n>6$  \cite{Sin14}}& 
$\left\{\begin{array}{l l}{{2^{n-1}-\frac{1}{2}\sqrt{(2^n-1)\sqrt{2^{\frac{3n}{2}+3}+2^{n+1}-2^{\frac{n}{2}+4}}+2^n}},}&{\text{if n is even} }
\\ {{= 2^{n-1}-2^{\frac{7n-2}{8}}-O(2^{\frac{3n}{8}})}}
\\ {{2^{n-1}-\frac{1}{2}\sqrt{(2^n-1)\sqrt{2^{\frac{3n+5}{2}}+2^{n+1}-2^{\frac{n+7}{2}}}+2^n}}},&{\text{if n is odd}}
\\ {{= 2^{n-1}-2^{\frac{7n-3}{8}}-O(2^{\frac{3n}{8}})}}
\end{array}\right.$ 
  \\ \hline
\end{tabular}
\label{funcs_third}
\end{table}

Garg and Khalyavin \cite{GK12} proved that the $r$-th order nonlinearity for the Kasami function $f(x)=\tr_n(\lambda x^k)$, where $k=2^{2r}-2^r+1$, $\lambda \in \mathbb{F}_{2^n}^*$, $n\ge 2r$ and $\gcd(n,r)=1$, is bounded by
\begin{equation*}
    \begin{cases}
        2^{n-r}-2^{\frac{n+2r-2}{2}},&\mbox{for even $n$}\\
        2^{n-r}-2^{\frac{n+2r-3}{2}},&\mbox{for odd $n$}
    \end{cases}
    .
\end{equation*}
Garg \cite{Gar15} proved that the $(\frac{n}{2}-1)$-th order nonlinearity of $\tr_n(\lambda x^{2^{\frac{n}{2}-1}})$ for $\lambda\in \mathbb{F}_{2^n}^*$ is at least $2^{\frac{n}{2}}$.
Tiwari and Sharma \cite{TS23} proved that the $(\frac{n}{2}-1)$-th order nonlinearity of $\tr_n(\lambda x^d)$, where $\lambda\in \mathbb{F}_{2^n}^*$ and $d=3(2^{\frac{n}{2}}-1)+1$ for even $n$, is at least $2^{\frac{n}{2}+1}-2^{\frac{\frac{n}{2}+1}{2}}$; the $(\frac{n}{2}-2)$-th order nonlinearity of $\tr_n(\lambda x^d)$, where $d=2^{\frac{n}{2}}-2$, is at least $ 2^{\frac{n}{2}+2}-2^{\frac{n}{4}+\frac{3}{2}}$. Saini and Garg \cite{SG22} proved that the $\frac{n}{4}$-th order nonlinearity of functions $\tr_n(\alpha_1x^{d_1}+\alpha_2x^{d_2})$ is at least $2^{\frac{3n}{4}}-2^{\frac{3n}{4}-2}$, where $\alpha_1,\alpha_2\in \mathbb{F}_{2^n}$, $d_1=\frac{1}{2}\cdot (2^{\frac{n}{2}}-1)+1$, $d_2=\frac{1}{6} \cdot (2^{\frac{n}{2}}-1)+1$, and $4\mid n$.

Proving large high-order nonlinearity lower bound for any explicit function is an outstanding open problem in the computational complexity. For example, the problem whether there exists a function in NP with $\mathrm{log}_2 n$-th order nonlinearity at least $2^{n-1}(1-\frac{1}{\sqrt{n}}) $ is open \cite{Vio22}. For the majority and mod functions, Razborov and Smolensky \cite{Raz87, Smo87, Smo93} proved that the $r$-th order nonlinearities of them are at least $2^{n-1}(1-O(\frac{r}{\sqrt{n}}))$. For $r\ll \log n$, Babai, Nisan and Szegedy \cite{BNS92} proved that the generalized inner product function has $r$-th order nonlinearity bounded by $2^{n-1}(1-\mathrm{exp}(-\Omega(\frac{n}{r\cdot 4^r})))$. Bourgain \cite{Bou05} proved a similar result for $\textrm{mod}_3$ function; a mistake in his proof is corrected by Green, Roy and Straubing \cite{GRS05}. An improvement was achieved by Viola and Widgerson \cite{Vio06,VW08} by exhibiting a polynomial-time computable function with $r$-th order nonlinearity lower bounded by $ 2^{n-1}(1-\mathrm{exp}(-\frac{\alpha\cdot n}{2^r}))$, where constant $\alpha< \frac{1}{4}\cdot \mathrm{log}_2 e$. Gopalan, Lovett and Shpilka \cite{GLS09} proved that, if the mod-$p$ degree, for any prime $p>2$, of $f$ is $d = o(\log n)$, then the $r$-th order nonlinearity of $f$ is at least $2^{n-1}(1-p^{-O(d)})$. Chattopadhyay \emph{et al.} \cite{CHHLZ20} proved that the $O(1)$-th nonlinearity for the $k$ XORs of the majority function is lower bounded by $2^{kn-1}\left(1-\left(\frac{\mathrm{poly}(k, \log n)}{\sqrt{n}}\right)^k\right)$. Chen and Lyu \cite{Chen21} proved that there exists a function $f\in \text{E}^{\text{NP}}$ which has $r$-th order nonlinearity at least $2^{n-1}(1-2^{-r})$ for $r\le o(\frac{n}{\mathrm{log}n})^{\frac{1}{2}}$.

\subsection{Our results}

In this work, we prove lower bounds on the high-order nonlinearities of certain trace monomial Boolean functions. We exhibit some trace monomial functions with large second-order, third-order or higher-order nonlinearities.

\begin{theorem}\label{theorem_1}
        Let $f(x)=tr_n(x^7)$. For even $n$, we have
    \begin{eqnarray*}
    \nl_2(f)&\ge &
        \begin{cases}
            2^{n-1}-\frac{1}{2}\sqrt{\frac{13}{3}\cdot 2^{\frac{3}{2}n-1}+2^n-\frac{1}{3}\cdot 2^{\frac{n}{2}+3}},& 3\nmid n
           \\ 2^{n-1}-\frac{1}{2}\sqrt{\frac{13}{3}\cdot 2^{\frac{3}{2}n-1}+2^{n+2}-\frac{1}{3}\cdot 2^{\frac{n}{2}+3}},& 3\mid n
        \end{cases} 
        \\ &=& 2^{n-1}-2^{\frac{3n}{4}-\frac{3}{2}+\frac{1}{2}\log_2 13- \frac{1}{2}\log_2 3}-O(2^{\frac{n}{4}})
        .
    \end{eqnarray*}
    For odd $n$, we have
    \begin{eqnarray*}
    \nl_2(f)&\ge &
        \begin{cases}
            2^{n-1}-\frac{1}{2}\sqrt{3\cdot 2^{\frac{3n-1}{2}}+2^n-2^{\frac{n+3}{2}}} ,& 3\nmid n
           \\ 2^{n-1}-\frac{1}{2}\sqrt{3\cdot 2^{\frac{3n-1}{2}}+2^n +3\cdot 2^{n+\frac{1}{2}}-2^{\frac{n+3}{2}}},& 3\mid n
        \end{cases}  
        \\ &=& 2^{n-1}-2^{\frac{3n-5}{4}+\frac{1}{2}\log_2 3}-O(2^{\frac{n}{4}})
        .
    \end{eqnarray*}
\end{theorem}

Theorem \ref{theorem_1} gives a lower bound on the second-order nonlinearity of $\tr_n(x^7)$. Among all trace monomials, it matches the best lower bound when $n$ is odd (i.e., the modified Welch function \cite{Car08}).

\begin{theorem}\label{theorem_2}
        Let $f(x)=\tr_n(x^{2^{r}+3})$, where $n=2r$. Then we have
     \begin{eqnarray*}
    \nl_2(f)&\ge &
        \begin{cases}
             2^{n-1} - \frac{1}{2}\sqrt{2^{\frac{3n}{2}+1}+ 2^{\frac{5n}{4}+\frac{1}{2}} - 2^n  - 2^{\frac{3n}{4}+\frac{1}{2}}},& 2\nmid r
            \\ 2^{n-1}-\frac{1}{2}\sqrt{2^{\frac{3}{2}n+1}+\frac{1}{3}\cdot 2^{\frac{5}{4}n+2}-2^n-\frac{1}{3}\cdot 2^{\frac{3}{4}n+2}},& 2\mid r
        \end{cases}
        \\ &=& 2^{n-1}-2^{\frac{3n}{4}-\frac{1}{2}}-O(2^{\frac{n}{2}}).
    \end{eqnarray*}
\end{theorem}

Theorem \ref{theorem_2} gives a lower bound on the second-order nonlinearity of $\tr_n(x^{2^{r}+3})$, where $n=2r$. When $n$ is even, it matches the largest lower bound on the second-order nonlinearities among all trace monomial Boolean functions. That is, it is the same as functions $\tr_n(x^{2^{r+1}+3})$, where $n=2r$ \cite{YT20}.
Note that a larger lower bound is known for Maiorana-McFarland type functions.
Kolokotronis and Limniotis proved that, the second-order nonlinearity for a cubic Maiorana-McFarland type functions $g(x)y^t$, where $(x,y)\in \mathbb{F}_{2^n} \times \mathbb{F}_{2^m}$ and $g(x)$ is a quadratic perfect nonlinear function, and $m\le \frac{n}{2}$,  is at least $2^{n+m-1}-2^{n-1}-2^{\frac{n}{2}+m-1}+2^{\frac{n}{2}-1}$ \cite{KL11}.

We would like to point out that, this class of functions $\tr_n(x^{2^{r}+3})$, where $n = 2r$, is studied for the first time in our work.
A similar type of functions $\tr_n(x^{2^{r+1}+3})$, where $n=2r$, was studied in \cite{CD96}; the lower bound proved in \cite{YT20} is exactly the same as Theorem \ref{theorem_2}.

\begin{theorem}\label{theorem_3}
    Let $f=\tr_n(x^{15})$. Then we have
   \begin{equation*}
    \nl_3(f)\ge 
        \begin{cases}
            2^{n-1}-\frac{1}{2} \sqrt{(2^n-1)\sqrt{\frac{1}{3}\cdot 2^{\frac{3}{2}n+4}+\frac{7}{3}\cdot 2^{n+1}-\frac{1}{3}\cdot 2^{\frac{n}{2}+5} }+2^n}\\ 
            = 2^{n-1}- 2^{\frac{7n}{8}-\frac{1}{4} \log_2 3}-O(2^{\frac{3n}{8}}), & 2\mid n
            \\
            \\   2^{n-1}-\frac{1}{2} \sqrt{(2^n-1)\sqrt{\frac{29}{8}\cdot 2^{\frac{3n+1}{2}}+2^{n+1}-7\cdot 2^{\frac{n+5}{2}}}+2^n}
            \\  = 2^{n-1}- 2^{\frac{7n}{8}-\frac{13}{8}+ \frac{1}{4} \log_2 29}-O(2^{\frac{3n}{8}}), &2\nmid n
        \end{cases}
    \end{equation*}
for $n\ge 6$.
\end{theorem}
Theorem \ref{theorem_3} gives a lower bound on the third-order nonlinearity for the functions within $\tr_n(x^{15})$ class; it is the largest lower bound on the third-order nonlinearity among all trace monomial Boolean functions. 

\begin{theorem}\label{theorem_4}
Let $f=\tr_n(x^{2^{r+1}-1})$ and $r\ge 2$.
    \begin{equation*}
        \begin{aligned}
            \nl_r(f)\ge 2^{n-1}-2^{(1-2^{-r})n+\frac{r}{2^{r-1}}-1}- O(2^{\frac{n}{2}}).
        \end{aligned}
    \end{equation*}
\end{theorem}

For $r \ll \log_2 n$, our lower bound in Theorem \ref{theorem_4} is better than all previous results, for all explicit functions in $\text{P}$, not necessarily trace monomials. 

Similarly, we prove the following lower bound on the $r$-th order nonlinearity for the inverse function, which is studied in \cite{Car08}. We credit this to Carlet, who claims that the $r$-th order nonlinearity for the inverse function is asymptotically lower bounded by $2^{n-1}-2^{(1-2^{-r})n}$.

\begin{theorem}\label{thm:inverse_lb}
    Let $f_{\mathrm{inv}}=\tr_n(x^{2^n-2})$. For any $r\ge 1$, we have $\nl_r(f_{\mathrm{inv}})\ge 2^{n-1}-2^{(1-2^{-r})n-2^{-(r-1)}}-O(2^{\frac{n}{2}})$.
\end{theorem}

\vspace{0.2cm}
\textbf{Techniques.} Our proof of the lower bounds follows from Carlet's methods \cite{Car08}. That is, to lower bound the $r$-th order nonlinearity, we estimate the (first-order) nonlinearity of its $(r-1)$-th order derivatives. Taking a (nontrivial) $(r-1)$-th order derivative, our target function becomes a quadratic function. Then, we rely on a result by Canteaut \emph{et al.} \cite{CCK08} that relates the nonlinearity of a quadratic function with the dimension of its \emph{linear kernel}. As such, the problem essentially reduces to  estimating \emph{the number of roots} of certain equations over the finite field $\mathbb{F}_{2^n}$.

As for Theorem \ref{theorem_1}, we use the following ingredients to estimate the number of roots of a certain equation (associated with the linear kernel): we factor the equation into irreducible ones; we apply the known results concerning the number of roots of \emph{$q$-polynomials}, and the number of roots of \emph{quadratic equations} and \emph{quartic equations} (over finite fields); we use the Weil bound to estimate the weight of trace monomial functions.
As for Theorem \ref{theorem_2}, our proof is similar to \cite{YT20}, and the lower bounds are exactly the same. (The target function, which has a simple form and good behavior, is somehow missed by previous works.)

As for the third-order nonlinearity lower bound, i.e., Theorem \ref{theorem_3}, our strategy is, again, to estimate the number of roots of a certain equation (associated with the linear kernel). We factor the equation into irreducible ones, and analyze the number of roots for each component separately. The proof relies on the known results about the number of roots of $q$-polynomials, and \emph{quartic equations} (over finite fields). A critical step is to estimate the algebraic degree of a (trace) equation over $\mathbb{F}_{2^n}$. (With the algebraic degree known, we can apply the well-known fact that the number of roots is bounded by the degree.)

In Theorem \ref{theorem_4}, we study the $r$-th order nonlinearity of functions $\tr_n(x^{2^{r+1}-1})$, a natural generalization of $\tr_n(x^{7})$ and $\tr_n(x^{15})$. We prove a lower bound on the (first-order) nonlinearity of all nontrivial $(r-1)$-th order derivatives of the target function, and the $r$-th order nonlinearity lower bound follows from the methods articulated by \cite{Car08}. The equation (associated with the linear kernel for the derivative) turns out to have a nice explicit form, whose degree is at most $2^{2r}$. Thus, the nonlinearity bound follows from a result in \cite{CCK08} (that relates the dimension of the kernel with the nonlinearity for any quadratic function).

The proof of Theorem \ref{thm:inverse_lb} closely follows from \cite{Car08}, who already claimed that the lower bound is asymptotically  $2^{n-1}-2^{(1-2^{-r})n}$. We credit the result to Carlet, who obviously can, but did not have the occasion to write down the details.

\section{Preliminary}
Let $\mathbb{F}_2$ be the finite field of size 2. Let $\mathcal{B}_n$ denote the set of all $n$-variable Boolean functions. Any $n$-variable Boolean function can be represented as a unique polynomial in $\mathbb{F}_2[x_1,x_2,\ldots,x_n]/\{x_i^2+x_i\}_{1\le i\le n}$, that is,
\[
f(x_1,x_2,\ldots,x_n)=\sum_{S\subseteq [n]}c_S \prod_{i\in S}x_i,
\]
which is called \emph{algebraic normal form} (ANF). The \emph{algebraic degree} of $f$, denoted by $\mathrm{deg}(f)$, is the number of variables in the highest order term with nonzero coefficient. 

The \emph{Hamming weight} of a vector $x \in \mathbb{F}_2^n$, denoted by $\wt(x)$, is the number of nonzero coordinates. The \emph{weight} of a Boolean function $f$, denoted by $\wt(f)$, is the cardinality of the set $\{ x \in \mathbb{F}_2^n : f(x) = 1\}$. The \emph{distance} between two functions $f$ and $g$ is the cardinality of the set $\{ x \in \mathbb{F}_2^n : f(x) \not= g(x)\}$, denoted by $\mathrm{d}(f, g)$.

Let $\mathbb{F}_{2^n}$ be the finite field of size $2^n$. The \emph{absolute trace function} from $\mathbb{F}_{2^n}$ to $\mathbb{F}_2$ can be defined as
\[
\tr_n(x)=x+x^2+x^{2^2}+\ldots+x^{2^{n-1}},
\]
where $x\in \mathbb{F}_{2^n}$. Let $ K=\mathbb{F}_{2^r}$ be a subfield of $ L=\mathbb{F}_{2^n}$. More generally, the trace function defined with respect to the field extension $L/K$ is
\[
\mathrm{tr}_{L/K}(\alpha)=\alpha+{\alpha}^{2^r}+\ldots+{\alpha}^{{2^{r\cdot(\frac{n}{r}-1)}}},
\]
where $\alpha\in \mathbb{F}_{2^n}$. It is well known that (for instance, Theorem 2.23 in \cite{LN97}) the trace function satisfies the following properties
\begin{itemize}
    \item $\tr_{L/K}(x+y)=\tr_{L/K}(x)+\tr_{L/K}(y)$ for any $x,y\in \mathbb{F}_{2^n}$.
    \item $\tr_{L/K}(x^2)=\tr_{L/K}(x)$ for any $x \in \mathbb{F}_{2^n}$.
    \item For any $\alpha \in \mathbb{F}_{2^r}$, there are exactly $2^{n-r}$ elements $\beta$ with $\mathrm{tr}_{L/K}(\beta) = \alpha$.
\end{itemize}

Any $n$-variable Boolean function can be written as $f(x)=\tr_n(g(x))$, where $g(x)=\sum_{i=0}^{2^n-1} \beta_i x^i$ is a mapping from $\mathbb{F}_{2^n}$ to $\mathbb{F}_{2^n}$ for $\beta_i\in \mathbb{F}_{2^n}$. A trace \emph{monomial} Boolean function is of the form $\mathrm{tr}_n(\lambda x^d)$ where $\lambda \in \mathbb{F}_{2^n}^*$ and $d$ is an integer. It is well known that the degree of the trace monomial function $\mathrm{tr}_n(\lambda x^d)$ is the Hamming weight of the binary representation of $d$ \cite{Car21}.

For $1 \le r \le n$, the $r$-th order nonlinearity of an $n$-variable Boolean function $f$, denoted by $\mathrm{nl}_r(f)$, is the minimum distance between $f$ and functions with degree at most $r$, i.e.,
\[
\mathrm{nl}_r(f) = \min_{\deg(g) \le r} \mathrm{d}(f, g).
\]
We denote by $\mathrm{nl}(f)$ the first-order nonlinearity of $f$. 

The \emph{Walsh transform} of $f\in \mathcal{B}_n$ at $\alpha \in \mathbb{F}_{2^n}$ is defined as
\[
W_f(\alpha)=\sum_{x\in \mathbb{F}_{2^n}}(-1)^{f(x)+\tr_n(\alpha x)}.
\]
The \emph{Walsh spectrum} of $f$ is the multi-set consisting of the values $W_f(\alpha)$ for all $\alpha \in \mathbb{F}_{2^n}$. The nonlinearity of any Boolean function in $n$ variable can be calculated as 
\begin{equation}\label{nonlinearity_walsh}
\mathrm{nl}(f)= 2^{n-1}-\frac{1}{2}\max_{\alpha \in \mathbb{F}_{2^n}} |W_f(\alpha)|.
\end{equation}

We denote by $D_{a}f$ the \emph{derivative} of the $f\in \mathcal{B}_n$ with respect to $a \in \mathbb{F}_{2^n}$, which is defined to be 
\[
D_af(x)=f(x)+f(x+a).
\]
The \emph{$k$-th order derivative} of $f$, denoted by $D_{a_1}D_{a_2}\ldots D_{a_k}f$, is obtained by applying such derivation successively to the function $f$ with respect to $a_1,a_2,\ldots,a_k \in \mathbb{F}_{2^n}$.

In \cite{Car08}, Carlet provided a method to lower bound the $r$-th order nonlinearity relying on the $(r-1)$-th order nonlinearity of all its derivatives. 

\begin{proposition}\label{lower_non} \cite{Car08} 
     Let $f$ be any $n$-variable Boolean function and $r$ a positive integer smaller than $n$. We have
     \[
        \nl_r(f)\ge 2^{n-1}-\frac{1}{2}\sqrt{2^{2n}-2\sum_{a\in \mathbb{F}_{2^n}}\nl_{r-1}(D_af)}.
    \]
\end{proposition}

The \emph{quadratic functions} are the set of the Boolean functions of algebraic degree at most 2. The \emph{linear kernel} is the central object for the calculation of the nonlinearity of quadratic functions.

\begin{definition}\label{kernel} \cite{CCK08}
    Let $q:\mathbb{F}_{2^n}\to \mathbb{F}_2$ be a quadratic function. The linear kernel of $q$, denoted by $\mathcal{E}_q$, can be defined as
    \[
    \mathcal{E}_q=\mathcal{E}_0\cup \mathcal{E}_1
    \]
    where 
    \[
    \mathcal{E}_0=\{b\in \mathbb{F}_{2^n}\mid D_bq=q(x)+q(x+b)=0,\ \mbox{for\ all\ $x\in \mathbb{F}_{2^n}$}\},
    \]
    \[
    \mathcal{E}_1=\{b\in \mathbb{F}_{2^n}\mid D_bq=q(x)+q(x+b)=1,\ \mbox{for\ all\ $x\in \mathbb{F}_{2^n}$}\}.
    \]
\end{definition}
The \emph{bilinear form} associated with a quadratic function $q$ is defined as 
\begin{equation*}
    B(x,y)=q(0)+q(x)+q(y)+q(x+y).
\end{equation*}

   The \emph{linear kernel} $\mathcal{E}_q$ of a quadratic function $q$ is the \emph{linear kernel} of its associated bilinear form $B(x,y)$ by definition, that is
\[
\mathcal{E}_q=\{x\in\mathbb{F}_{2^n}\mid B(x,y)=0\ \mbox{for\ any\ $y\in \mathbb{F}_{2^n}$}\}.
\] 

\begin{lemma}\cite{CCK08} \label{Walsh_spec_quar}
Let $q:\mathbb{F}_{2^n}\to \mathbb{F}_2$ be an $n$-variable Boolean function of degree at most 2. Then the Walsh spectrum of $q$ depends on the dimension $k$ of the \emph{linear kernel} of $q$. Moreover, for any $\mu \in \mathbb{F}_{2^n}$, we have
\begin{table}[H]
	\begin{center}
 \renewcommand\arraystretch{1.5}
		\begin{tabular}{ c|c} 
			\hline 
			$W_q(\mu)$ & The number of $u\in \mathbb{F}_{2^n}$ \\ 
			\hline
			0 & $2^n-2^{n-k}$   \\ 
			\hline
			$2^{\frac{n+k}{2}}$  & $2^{n-k-1}+(-1)^{q(0)}2^{\frac{n-k-2}{2}}$  \\ 
			\hline
			$-2^{\frac{n+k}{2}}$ & $2^{n-k-1}-(-1)^{q(0)}2^{\frac{n-k-2}{2}}$ \\ 
			\hline
		\end{tabular}
	\end{center}
\end{table}

\end{lemma}

\begin{lemma} \cite{CCK08} \label{parity_dim}
Let $V$ be a vector space over a field $\mathbb{F}_{2^n}$ and $Q:V\to \mathbb{F}_{2^n}$ be a quadratic form. Then the dimension of $V$ and the dimension of the kernel of $Q$ have the same parity.
\end{lemma}
That is, if $f:\mathbb{F}_{2^n}\to \mathbb{F}_2$ is a quadratic function, then the parity of the dimension of its linear kernel is the same as the parity of $n$.

A \emph{q-polynomial} over $\mathbb{F}_{q^n}$ is the polynomial in the form  
\[
P(x)=\sum_{i=0}^{n-1}a_ix^{q^i},
\] 
where the coefficients $a_i \in \mathbb{F}_{q^n}$. It is a \emph{linearized polynomial} which satisfies the following properties \cite[page 108]{LN97}:
\begin{equation}\label{proper_1}
         P(b+c)=P(b)+P(c),\ \ \ \text{for\ all\ $b,c\in \mathbb{F}_{q^n}$}
\end{equation}

\begin{equation}\label{proper_2}
  P(tb)=tP(b),\ \ \ \text{for\ all\ $t\in \mathbb{F}_q$,\ all\ $b\in \mathbb{F}_{q^n}$}.
\end{equation}
Equation \eqref{proper_1} follows from the fact that $(a+b)^{q^i}=a^{q^i}+b^{q^i}$ for $a,b\in \mathbb{F}_{q^n}$ and $i\ge 0$ \cite[Theorem 1.46]{LN97}; equation \eqref{proper_2} follows from that $t^{q^i}=t$ for $t\in \mathbb{F}_q$ and any $i\ge 0$. Hence, if $\mathbb{F}_{q^n}$ is regarded as a vector space over $\mathbb{F}_q$, then a $q$-polynomial is a linear map of this vector space.

\section{Second-order nonlinearity}

In this section, we deduce that the lower bound on the second-order nonlinearity for two classes of trace monomial Boolean functions in the form $\tr_n(x^7)$ and $\tr_n(x^{2^{r}+3})$, where $n=2r$.

\subsection{The functions $\tr_n(x^7)$}

We will lower bound the second-order nonlinearity of the monomial cubic functions $\tr_n(x^7)$. The algebraic degree of the derivatives of $\tr_n(x^7)$ is at most 2 since the degree of $\tr_n(x^7)$ is exactly 3. By using Carlet's method (i.e., Proposition \ref{lower_non}), our goal is to calculate the nonlinearities of all its derivatives. 

\begin{proposition}\label{equivalent}
Let $f:\mathbb{F}_{2^n}\to \mathbb{F}_2$ be a quadratic function. For any $a \in \mathbb{F}_{2^n}^*$, we have 
\[
\mathcal{E}_f=\mathcal{E}_{f(ax)},
\]
where $\mathcal{E}_f$ denotes the linear kernel of the $f$ and $\mathcal{E}_{f(ax)}$ denotes the linear kernel of the $f(ax)$.
\end{proposition}
\begin{proof}
  Let us prove $\mathcal{E}_{f}\subseteq \mathcal{E}_{f(ax)}$ first. By definition, if $b\in \mathcal{E}_{f}$, then $f(x)+f(x+b)=0$ for all $x\in \mathbb{F}_{2^n}$ or $f(x)+f(x+b)=1$ for all $x\in \mathbb{F}_{2^n}$. Note that $x \mapsto ax$ is a bijection over $\mathbb{F}_{2^n}$ for any $a\in \mathbb{F}_{2^n}^*$, then we have  
  \[
  f(ax)+f(ax+b)=0\ \mbox{for\ all\ $x\in \mathbb{F}_{2^n}$ } 
  \]
  or
  \[
  f(ax)+f(ax+b)=1\ \mbox{for\ all\ $x\in \mathbb{F}_{2^n}$ } .
  \]
  So $b\in \mathcal{E}_{f(ax)}$. 

Now let us prove $\mathcal{E}_{f(ax)}\subseteq \mathcal{E}_{f}$.
Let $g(x) = f(ax)$. From the above, we have $\mathcal{E}_g \subseteq \mathcal{E}_{g(a^{-1}x)}$, that is, $ \mathcal{E}_{f(ax)} \subseteq \mathcal{E}_{f}$.
 
\end{proof}

We will need the following lemmas in the proof of Theorem \ref{dimension_kernel}. 

\begin{lemma} \label{root_number1} \cite[Theorem 3.50]{LN97}
Let $q$ be a prime. Let $P(x)=\sum_{i=0}^{n-1}a_i x^{q^i}$ be a $q$-polynomial, where $a_i\in \mathbb{F}_{q^n}$. Then the distinct number of roots of $P(x)$ in $\mathbb{F}_{q^n}$ is a power of $q$.
\end{lemma}

\begin{lemma} (\cite[page 37]{Men93}) \label{num_root_quar} 
    The number of solutions in $\mathbb{F}_{2^n}$ of the quartic function
    \begin{equation}\label{quar_equ}
        x^4+ax+b=0,\ \ a,b\in \mathbb{F}_{2^n},\ a\neq 0.
    \end{equation}
    \begin{itemize}
        \item If $n$ is odd, then \eqref{quar_equ} has either no solution or exactly two solutions.
        \item If $n$ is even and $a$ is not a cube, then \eqref{quar_equ} has exactly one solution.
        \item If $n$ is even, and $a$ is a cube, then \eqref{quar_equ} has four solutions if $\mathrm{tr}_{\mathbb{F}_{2^n}/\mathbb{F}_{4}}(\frac{b}{a^{\frac{4}{3}}})=0$, and no solutions if $\mathrm{tr}_{\mathbb{F}_{2^n}/\mathbb{F}_{4}}(\frac{b}{a^{\frac{4}{3}}})\neq 0$.
    \end{itemize}
\end{lemma}

We need some properties of trace functions in the proof.

\begin{theorem}\label{dimension_kernel}
Let $f(x)=\tr_n(x^7)$. Let $\mathcal{E}_{D_af}$ be the linear kernel of $D_af$. We denote by $\dim(\mathcal{E}_{D_af})$ the dimension of $\mathcal{E}_{D_af}$. The distribution of $\dim(\mathcal{E}_{D_af})$ for all $a\in \mathbb{F}_{2^n}^*$ is as follows:
\begin{table}[H]
\begin{center}
\caption{The distribution of $\dim(\mathcal{E}_{D_af})$} 
\renewcommand\arraystretch{1.5}
\begin{tabular}{ c|c|c|c } 
\hline 
\multicolumn{2}{c|}{n} & $\dim(\mathcal{E}_{D_af})$ & The number of $a\in \mathbb{F}_{2^n}^*$ \\ 
\hline
\multirow{4}*{even $n$} & \multirow{2}*{$3\nmid n$} & 2 & $\frac{11}{3}\cdot 2^{n-2}-\frac{2}{3}$ \\
                \cline{3-4}
                ~ & ~ & 4 & $\frac{1}{3} \cdot 2^{n-2} -\frac{1}{3}$   \\ 
                \cline{2-4}
                ~ & \multirow{2}*{$3\mid n$} & 2 & $\frac{2}{3}(2^n-1)+\frac{1}{2} \wt(\tr_n(x^7))$ \\
                \cline{3-4}
                ~ & ~ & 4 & $\frac{1}{3} (2^n-1) -\frac{1}{2} \wt(\tr_n(x^7))$ \\
                \hline
\multirow{4}*{odd $n$} & \multirow{2}*{$3\nmid n$} & 1 & $2^{n-1}$ \\
                \cline{3-4}
                ~ & ~ & 3 & $2^{n-1}-1$   \\ 
                \cline{2-4}
                ~ & \multirow{2}*{$3\mid n$} & 1 & $\wt(\tr_n(x^7))$ \\
                \cline{3-4}
                ~ & ~ & 3 & $2^n-1-\wt(\tr_n(x^7))$ \\
			\hline
		\end{tabular}
		
	\end{center}
\end{table}

\end{theorem}

\begin{proof}
For any $a\in \mathbb{F}_{2^n}^*$, we have 
\begin{eqnarray*}
  (D_af)(ax)&= & \tr_n({(ax)}^7) + \tr_n({(ax+a)}^7)\\
  & = &  \tr_n({(ax)}^7+ {(ax+a)}^7)\\
  & = &  \tr_n(a^7(x^6+x^5+x^4+x^3+x^2+x+1)).
\end{eqnarray*}
Let $g(x)=D_af(ax)$. By Proposition \ref{equivalent}, we know $\mathcal{E}_{g}=\mathcal{E}_{D_af}$, and $\dim(\mathcal{E}_{g})$ equals the number of $b\in \mathbb{F}_{2^n}$ such that $D_b g$ is a constant. 
Note that
\begin{eqnarray*}
D_b g(x) & = & \tr_n(a^7 (\sum_{i=0}^6 x^i ))  + \tr_n(a^7 (\sum_{i=0}^6 (x+b)^i )) \\
& = & \tr_n(a^7((b^2+b)x^4+(b^4+b)x^2+(b^4+b^2)x)) + \tr_n(\sum_{i=1}^6 b^i).
\end{eqnarray*}
So $\dim(\mathcal{E}_{g})$ equals the number of $b\in \mathbb{F}_{2^n}$ such that $\tr_n(a^7((b^2+b)x^4+(b^4+b)x^2+(b^2+b^4)x))$ is a constant.

Using the properties of the trace function, we have
\begin{eqnarray}
& &\tr_n(a^7((b^2+b)x^4+(b^4+b)x^2+(b^4+b^2)x)) \nonumber\\
&=&\tr_n(a^7((b^2+b)x^4))+\tr_n(a^7((b^4+b)x^2))+\tr_n(a^7((b^4+b^2)x)) \nonumber\\
&=&\tr_n((a^7)^{-4}((b^2+b)^{-4}x))+ \tr_n((a^7)^{-2}((b^4+b)^{-2}x))+\tr_n(a^7((b^4+b^2)x)) \nonumber\\
&=& \tr_n(((a^7)^{-4}(b^2+b)^{-4}+(a^7)^{-2}(b^4+b)^{-2}+a^7(b^4+b^2))x). \label{equ:thm1_tr_constant}
\end{eqnarray}
Thus, \eqref{equ:thm1_tr_constant} is a constant if and only if the coefficient of $x$ is zero, that is, 
\begin{equation}
\label{equ:thm1_coef_x}
(a^7)^{-4}(b^2+b)^{-4}+(a^7)^{-2}(b^4+b)^{-2}+a^7(b^4+b^2)=0.    
\end{equation}
Taking the $4$th power to both sides of \eqref{equ:thm1_coef_x}, we have
\begin{eqnarray*}
0 & = & a^7(b^2+b)+a^{14}(b^4+b)^{2}+a^{28}(b^4+b^2)^4 \\
& = & a^{28}(b^2+b)^8+a^{14}(b^2+b)^4+a^{14}(b^2+b)^2+a^7(b^2+b) \\
& = & \left( (a^7)^2(b^2+b)^4 \right)^2+\left( (a^7)^2(b^2+b)^4\right)+\left( a^7(b^2+b) \right)^2+\left( a^7(b^2+b)\right) \\
& = & \left((a^7)^2(b^2+b)^4+ a^7(b^2+b)\right)\left((a^7)^2(b^2+b)^4+ a^7(b^2+b)+1\right).
\end{eqnarray*}
For convenience, let $P(a, b) = \left((a^7)^2(b^2+b)^4+ a^7(b^2+b)\right)\left((a^7)^2(b^2+b)^4+ a^7(b^2+b)+1\right) = Q(a,b) (Q(a,b) + 1)$, where $Q(a, b) =(a^7)^2(b^2+b)^4+ a^7(b^2+b)$.

We denote by $\solnum(a)$ the number of $b\in \mathbb{F}_{2^n}$ such that $P(a, b)=0$; denote by $\solnumone(a)$ the number of $b\in \mathbb{F}_{2^n}$ such that $Q(a, b) = 0$; denote by $\solnumtwo(a)$ the number of $b\in \mathbb{F}_{2^n}$ such that $Q(a, b) + 1=0$. Obviously, $\solnum(a) = \solnumone(a)+\solnumtwo(a)$.

It is clear that $b = 0, 1$ are two solutions of $Q(a, b) = 0$. If $b^2 + b \not= 0$, $Q(a, b) = 0$ is equivalent to
\begin{equation}\label{eq3}
    (b^2+b)^3=(a^7)^{-1}.
\end{equation}
Observe that the degree (in variable $b$) of the polynomial $P(a,b)$ is 16. So $\solnum(a) \le 16$. Since $b = 0$ or $1$ are two distinct roots of $Q(a, b) = 0$, we have $\solnum(a) \ge 2$. For any fixed $a \in \mathbb{F}_{2^n}^*$, note that $P(a,b)$ is a $2$-polynomial in variable $b$. By Lemma \ref{root_number1},  $\solnum(a)=2^k$ for some $1\le k \le 4$. By Lemma \ref{parity_dim}, we know that $\dim( \mathcal{E}_g)$ and $n$ have the same parity. Hence, we have $\solnum(a) \in \{2^2, 2^4\}$ when $n$ is even; $\solnum(a) \in \{2^1, 2^3\}$ when $n$ is odd.

Next, we will consider the cases according to the parity of $n$ to determine the distribution of $N(a)$, i.e., the distribution of $\dim(\mathcal{E}_{D_af})$.

\vspace{0.2cm}
\textbf{Case 1:} $n$ is even. In this case, $\solnum(a) \in \{2^2, 2^4\}$; it suffices to count the number of $a\in \mathbb{F}_{2^n}^*$ where $\solnum(a)=16$. Note that the degree of $Q(a, b)$, for any fixed $a\in \mathbb{F}_{2^n}^*$, is 8. So we have $\solnumone(a)\le 8$ and $\solnumtwo(a)\le 8$. Hence, $\solnum(a)=16$ if and only if $\solnumone(a)=\solnumtwo(a)=8$. 

For even $n$, we have $\gcd(2^n-1, 3) =3$ and $\gcd(2^n-2,3)=1$ since $2^n\equiv 1 \pmod 3$. Let $G=\{g^{3s}\mid 0\le s\le \frac{2^n-4}{3}\}$ be a multiplicative group of order $\frac{2^n-1}{3}$, where $g$ is a primitive element of $\mathbb{F}_{2^n}^*$. If $a^7\notin G$, there is no solution to \eqref{eq3}, which implies that $\mathrm{N}_1(a)=2$. If $a^7\in G$, letting $a^7=g^{3s}$, where $0\le s\le \frac{2^n-4}{3}$, we have 
\begin{equation}\label{rootequ}
b^2+b=g^{-s+\frac{(2^n-1)i}{3}},
\end{equation}
for $i = 0,1,2$. If $\solnumone(a) = 8$, then \eqref{rootequ} must have 2 solutions for each $i = 0, 1, 2$. As a result, $\tr_n(g^{-s+\frac{(2^n-1)i}{3}})=0$ must hold for each $i$. (It is known that $x^2+x=b$ has two solutions if and only if $\tr_n(b)=0$, for instance, see the theorem in \cite[page 536]{Car21}.)

Let $c=g^{-s} $ and $d=g^{\frac{2^n-1}{3}} $. We have $g^{2^n-1-s}=c$, $g^{\frac{2^n-1}{3}-s}=cd$ and $g^{\frac{2(2^n-1)}{3}-s}=cd^2$. Furthermore, we have
\begin{eqnarray*}
& & \tr_n(c) + \tr_n( cd)+ \tr_n(cd^2) \\
& = & \tr_n(c(1+d+d^2) \\
& = & \tr_n(c(1+d+d^2)(1+d)(1+d)^{-1}) \\
& = & \tr_n(c(1+d^3)(1+d)^{-1}) \\
& = & 0,
\end{eqnarray*}
since $d^3=g^{2^n-1}=1$. In other words, 
\begin{equation}\label{keypoint}
   \tr_n(g^{2^n-1-s})+\tr_n( g^{\frac{2^n-1}{3}-s})+\tr_n(g^{2\frac{2^n-1}{3}-s})=0 
\end{equation}
always holds for any $0\le s\le \frac{2^n-4}{3}$. By \eqref{keypoint}, there are two possibilities:
\begin{itemize}
    \item $\tr_n(g^{2^n-1-s})=\tr_n( g^{\frac{2^n-1}{3}-s})=\tr_n(g^{2\frac{2^n-1}{3}-s})=0$,
    \item $\tr_n(g^{\frac{(2^n-1)i_1}{3}-s})=\tr_n( g^{\frac{(2^n-1)i_2}{3}-s})=1$ and $\tr_n(g^{\frac{(2^n-1)i_3}{3}-s})=0$ for distinct $i_1,i_2,i_3\in\{0,1,2\}$.
\end{itemize}

To proceed, we consider the following two subcases.

\textbf{Subcase 1.1}. $3\nmid n$ and $n$ is even.
    In this case, $\gcd(2^n-1,7)=1$, so the linear function $a\mapsto a^7$ is a bijection from $\mathbb{F}_{2^n}$ to $\mathbb{F}_{2^n}$. Hence, the number of $a\in \mathbb{F}_{2^n}^*$ such that $N_1(a)=8$ is exactly the size of the set $\{0\le s\le \frac{2^n-4}{3} \mid \tr_n(g^{-s+\frac{(2^n-1)i}{3}})=0\ \text{for $i=0,1,2$}\}$. 

Denote $s_1$ by the size of the set $\{0\le s\le \frac{2^n-4}{3}\mid \tr_n(g^{2^n-1-s})=\tr_n( g^{\frac{2^n-1}{3}-s})=\tr_n(g^{\frac{2(2^n-1)}{3}-s})=0\}$; denote $s_2$ by the size of the set $\{0\le s\le \frac{2^n-4}{3}\mid \tr_n( g^{\frac{(2^n-1)i_1}{3}-s})=\tr_n(g^{\frac{(2^n-1)i_2}{3}-s})=1 \text{ and } \tr_n( g^{\frac{(2^n-1)i_3}{3}-s})=0 \text{\ for distinct }\ i_1,i_2,i_3\in \{0,1,2\}\}$. Observe that $\wt(\tr_n(x))=2^{n-1}$ because $\tr_n(x)$ is an affine function, and the set $\{g^{-s+\frac{(2^n-1)i}{3}}\mid 0\le s\le \frac{2^n-4}{3}, 0\le i\le 2\}$ is exactly $\mathbb{F}_{2^n}^*$. So we have 
\begin{equation}\label{eq123}
    \begin{cases}
        3(s_1+s_2)=2^n-1 \\
        2s_2=2^{n-1},
    \end{cases}
\end{equation}
where $2s_2=\wt(\tr_n(x))=2^{n-1}$ is because $2s_2$ is the weight of the function $\tr_n(x)$. Solving equations \eqref{eq123}, we have $s_1=\frac{2^{n-2}-1}{3}$ and $s_2=2^{n-2}$. Thus the number of $a\in \mathbb{F}_{2^n}^*$ such that $\solnumone(a)=8$ is $\frac{2^{n-2}-1}{3}$. Therefore, the number of $a\in \mathbb{F}_{2^n}^*$ such that $\solnum(a)=16$ is $\frac{2^{n-2}-1}{3}$ and the number of $a\in \mathbb{F}_{2^n}^*$ such that $\solnumtwo(a)=4$ is $\frac{11}{3}\cdot 2^{n-2}-\frac{2}{3}$.

\textbf{Subcase 1.2}. $3\mid n$ and $n$ is even.
In this case, we have $7\mid 2^n-1$; thus the function $a\mapsto a^7$ is a $7$-to-$1$ mapping from $\mathbb{F}_{2^n}^*$ to $\mathbb{F}_{2^n}^*$. So $\{a^7\mid a^7\in G\}=\{g^{3s}\mid 0\le s\le \frac{2^n-4}{3} \text{ and } 7\mid s\}$. Denote by $s_1$ the size of the set $\{ 0\le s\le \frac{2^n-4}{3} \mid \tr_n(g^{-s+\frac{(2^n-1)i}{3}})=0 \text{ for all } i=0,1,2, \text{ and } 7\mid s \}$ and by $s_2$ the size of the set $\{0\le s\le \frac{2^n-4}{3} \mid \tr_n( g^{\frac{(2^n-1)i_1}{3}-s})=\tr_n(g^{\frac{(2^n-1)i_2}{3}-s})=1 \text{ and } \tr_n( g^{\frac{(2^n-1)i_3}{3}-s})=0, \text{for distinct } i_1,i_2,i_3\in \{0,1,2\}, \text{ and } 7 \mid s\}$. One can easily verify that
\begin{equation} \label{odd_mid}
    \begin{cases}
        s_1+s_2=\frac{2^n-1}{21}, \\
        14s_2=\wt(\tr_n(x^7)).
    \end{cases}
\end{equation}
Solving equations \eqref{odd_mid}, we have $s_1 = \frac{2^n-1}{21}-\frac{\wt(\tr_n(x^7))}{14}$ and $s_2 = \frac{\wt(\tr_n(x^7))}{14}$.
Hence, the number of $a\in \mathbb{F}_{2^n}^*$ such that $\mathrm{N}_1(a)=8$ is $7s_1=\frac{2^n-1}{3}-\frac{\wt(\tr_n(x^7))}{2}$. The number of $a\in \mathbb{F}_{2^n}^*$ such that $\solnum(a)=16$ is $\frac{2^n-1}{3}-\frac{\wt(\tr_n(x^7))}{2}$, and the number of $a\in \mathbb{F}_{2^n}^*$ such that $\solnum(a)=4$ is $\frac{2}{3}(2^n-1)+\frac{\wt(\tr_n(x^7))}{2}$.

\vspace{0.2cm}
\textbf{Case 2:}  $n$ is odd. In this case, we have $\solnum(a) \in \{2^1, 2^3\}$; it suffices to count the number of $a\in \mathbb{F}_{2^n}^*$ such that $\solnum(a)=8$. For odd $n$, we have $3\mid(2^n-2)$ and $\gcd(3, 2^n-1) = 1$. So $a \mapsto a^3$ is a bijection in $\mathbb{F}_{2^n}^*$. By \eqref{eq3}, we have
 \begin{equation}\label{eq_b_quar}
     b^2+b=(a^7)^{\frac{2^n-2}{3}}.
 \end{equation}
 When $b \not\in \{0, 1\}$, equation \eqref{eq_b_quar} has two distinct solutions if and only if $\tr_n((a^7)^{\frac{2^n-2}{3}})=0$. Hence, the number of solutions of $Q(a, b) = 0$ is at most 4, i.e., $\solnumone(a)\le 4$.

Note that $Q(a,b)$ is a 2-polynomial (in variable $b$) of degree 8 and $b=0,1$ are two roots of $Q(a,b)=0$. So we have $\solnumone(a)\in \{2,2^2\}$. By Lemma \ref{num_root_quar}, for odd $n$, the number of distinct $b^2+b$ satisfying $Q(a,b)+1=0$ is 0 or 2. So the number of $b\in \mathbb{F}_{2^n}\setminus \{0,1\}$ such that $Q(a,b)+1=0$ is $0, 2, 4$, that is, $\solnumtwo(a)\in \{0,2,4\}$. Thus $\solnum(a)=\solnumone(a)+\solnumtwo(a)=8$ if and only if  $\solnumone(a)=4$.

\textbf{Subcase 2.1}: $3\nmid n$ and $n$ is odd. In this case, we have $\gcd(2^n-1,7)=1$. So mapping $a\mapsto a^7$ is a bijection from $\mathbb{F}_{2^n}$ to $\mathbb{F}_{2^n}$. Since $\gcd(2^n-1,\frac{2^n-2}{3})=1$. then mapping $a \mapsto a^{\frac{2^n-2}{3}} $ is a bijection from $\mathbb{F}_{2^n}$ to $\mathbb{F}_{2^n}$. Note that $\solnumone(a)=4$ if and only if $\tr_n((a^7)^{\frac{2^n-2}{3}})=0$. As such, the number of $a\in \mathbb{F}_{2^n}^*$ such that $\mathrm{N}_1(a)=4$ equals the size of the set $\{x\in \mathbb{F}_{2^n}^*\mid \tr_n(x)=0\}$. Note that $\tr_n(x)$ is an affine function. So the number of $x\in \mathbb{F}_{2^n}^*$ such that $\tr_n(x)=0$ is $2^{n-1}-1$. Thus the number of $a\in \mathbb{F}_{2^n}^*$ such that $\solnumone(a)=4$ equals $2^{n-1}-1$.

\textbf{Subcase 2.2}: $3\mid n$ and $n$ is odd. In this case, we have $\gcd(2^n-1,\frac{2^n-2}{3})=1$. So mapping $a \mapsto a^{\frac{2^n-2}{3}} $ is a bijection from $\mathbb{F}_{2^n}$ to $\mathbb{F}_{2^n}$. Since $\gcd(2^n-1,7)=7$, then $a\mapsto a^7$ is a $7$-to-$1$ mapping from $\mathbb{F}_{2^n}^*$ to $\mathbb{F}_{2^n}^*$. Hence, the number of $a\in \mathbb{F}_{2^n}^*$ such that $\tr_n((a^{7})^{\frac{2n-2}{3}})=\tr_n((a^{\frac{2n-2}{3}})^7)=0$ equals the number of $a\in \mathbb{F}_{2^n}^*$ such that $\tr_n(a^{7})=0$. Since $\solnumone(a)=4$ if and only if $\tr_n((a^7)^{\frac{2^n-2}{3}})=0$, then one can easily verify that the number of $a\in \mathbb{F}_{2^n}^*$ such that $\solnumone(a)=4$ is $2^n-1-\wt(\tr_n(x^7))$.

\end{proof}

By Lemma \ref{Walsh_spec_quar} and Theorem \ref{dimension_kernel}, the following corollary is immediate.

\begin{corollary}\label{der_nonlinearity}
    Let $f=\tr_n(x^7)$. Denote $\nl(D_af)$ by the nonlinearity of $D_af$. For any $a\in \mathbb{F}_{2^n}^*$, the distribution of $\nl(D_af)$ is as follows:

 \begin{table}[H]
\begin{center}
\caption{The distribution of $\nl(D_af)$} 
\renewcommand\arraystretch{2}

\begin{tabular}{ c|c|c|c } 
\hline 
\multicolumn{2}{c|}{n} & $\nl(D_af)$ & The number of $a\in \mathbb{F}_{2^n}^*$ \\ 
\hline
\multirow{4}*{even $n$} & \multirow{2}*{$3\nmid n$} & $2^{n-1}- 2^{\frac{n}{2}}$ & $\frac{11}{3}\cdot 2^{n-2}-\frac{2}{3}$ \\
                \cline{3-4}
                ~ & ~ & $2^{n-1}- 2^{\frac{n+2}{2}}$ & $\frac{1}{3} \cdot 2^{n-2} -\frac{1}{3}$    \\ 
                \cline{2-4}
                ~ & \multirow{2}*{$3\mid n$} & $2^{n-1}- 2^{\frac{n}{2}}$ & $\frac{2}{3}(2^n-1)+\frac{1}{2}\wt(\tr_n(x^7))$ \\
                \cline{3-4}
                ~ & ~ & $2^{n-1}- 2^{\frac{n+2}{2}}$ & $\frac{1}{3}(2^n-1)-\frac{1}{2} \wt(\tr_n(x^7))$ \\
                \hline
\multirow{4}*{odd $n$} & \multirow{2}*{$3\nmid n$} & $2^{n-1}- 2^{\frac{n-1}{2}}$ & $2^{n-1}$ \\
                \cline{3-4}
                ~ & ~ & $2^{n-1}- 2^{\frac{n+1}{2}}$ & $2^{n-1}-1$   \\ 
                \cline{2-4}
                ~ & \multirow{2}*{$3\mid n$} & $2^{n-1}- 2^{\frac{n-1}{2}}$ & $\wt(\tr_n(x^7))$ \\
                \cline{3-4}
                ~ & ~ & $2^{n-1}- 2^{\frac{n+1}{2}}$ & $2^n-1-\wt(\tr_n(x^7))$ \\
			\hline
		\end{tabular}
		
	\end{center}
\end{table}
    
\end{corollary}

\begin{theorem}(The Weil bound, for example, Theorem 5.38 in \cite{LN97})\label{Weil's Theorem}
     Let $f \in \mathbb{F}_q[x]$ be of degree $d\ge 1$, where $\gcd(d,q)=1$. Let $\mathcal{X}$ be a nontrivial additive character of $\mathbb{F}_q$. Then 
     \begin{equation*}
         \left| \sum_{x\in \mathbb{F}_{q}}\mathcal{X}(f(x))\right| \le (n-1)q^{\frac{1}{2}}.
     \end{equation*}
     
\end{theorem}

\begin{lemma}\label{Weil's bound}
    Let $d \ge 1$ be an odd number. We have $\wt(\tr_n(x^d))\ge 2^{n-1} - \frac{d-1}{2}\cdot2^{\frac{n}{2}}$.
\end{lemma}
\begin{proof}
       Let $\mathcal{X}(x)=e^{\frac{2\pi i  \tr_n(x)}{p}}=  (-1)^{\tr_n(x)}$ for $p=2$. Applying the Weil bound, i.e., Theorem \ref{Weil's Theorem}, we have
    \begin{eqnarray*}
        \left|\sum_{x\in \mathbb{F}_{2^n}}\mathcal{X}(x^d) \right| 
        &=& \left|\sum_{x\in \mathbb{F}_{2^n}}(-1)^{\tr_n(x^d)}\right|
        \\ &\le& (d-1)2^{\frac{n}{2}}.
    \end{eqnarray*} 
    Since $\wt(\tr_n(x^d)) = 2^{n-1}-\frac{1}{2}\mid \sum_{x\in \mathbb{F}_{2^n}}(-1)^{\tr_n(x^d)}\mid $, we have 
    \[
    \wt(\tr_n(x^d))\ge 2^{n-1} - \frac{d-1}{2}\cdot 2^{\frac{n}{2}}.
    \]
\end{proof}

Now we are ready to prove Theorem \ref{theorem_1}, which gives a lower bound on the second-order nonlinearity of $\tr_n(x^7)$.

\begin{proof} (of Theorem \ref{theorem_1})
    By Proposition \ref{lower_non} and Corollary \ref{der_nonlinearity}, when $n$ is even and $3\nmid n$, we have 
    \begin{eqnarray*}
            \nl_2(f) & \ge & 2^{n-1}-\frac{1}{2}\sqrt{2^{2n}-2\sum_{a\in \mathbb{F}_{2^n}}\nl(D_af)}\\&
            = &2^{n-1}-\frac{1}{2}\sqrt{2^{2n}-2((2^{n-1}-2^{\frac{n}{2}})(\frac{11}{3}\cdot 2^{n-2}-\frac{2}{3})+(2^{n-1}-2^{\frac{n+2}{2}})(\frac{1}{3} \cdot 2^{n-2}-\frac{1}{3}))} \\&
            = &2^{n-1}-\frac{1}{2}\sqrt{ \frac{13}{3}\cdot2^{\frac{3}{2}n-1}+2^n-\frac{1}{3}\cdot 2^{\frac{n}{2}+3}}
            \\& = & 2^{n-1}-2^{\frac{3n}{4}-\frac{3}{2}+\frac{1}{2}\log_2 13- \frac{1}{2}\log_2 3}-O(2^{\frac{n}{4}}). 
    \end{eqnarray*}
    Similarly, when $n$ is even and $3\mid n$, we have 
\begin{eqnarray*}
    \nl_2(f) & \ge & 2^{n-1}-\frac{1}{2}\sqrt{\frac{1}{3}\cdot2^{\frac{3}{2}n+3}+2^n-\frac{1}{3}\cdot2^{\frac{n}{2}+3}-\wt(\tr_n(x^7))\cdot 2^{\frac{n}{2}}} \\
&\ge & 2^{n-1}-\frac{1}{2}\sqrt{\frac{13}{3}\cdot 2^{\frac{3}{2}n-1}+2^{n+2}-\frac{1}{3}\cdot2^{\frac{n}{2}+3}}\\
& = & 2^{n-1}-2^{\frac{3n}{4}-\frac{3}{2}+\frac{1}{2}\log_2 13- \frac{1}{2}\log_2 3}-O(2^{\frac{n}{4}}),
\end{eqnarray*}
where the second step is because $\wt(\tr_n(x^7))\ge 2^{n-1}- 3\cdot2^{\frac{n}{2}}$ by Lemma \ref{Weil's bound}.

By Proposition \ref{lower_non} and Corollary \ref{der_nonlinearity}, for odd $n$ and $3\nmid n$ we have 
    \begin{eqnarray*}
            \nl_2(f)&\ge & 2^{n-1}-\frac{1}{2}\sqrt{2^{2n}-2\sum_{a\in \mathbb{F}_{2^n}}\nl(D_af)}\\&
            = &2^{n-1}-\frac{1}{2}\sqrt{2^{2n}-2((2^{n-1}-2^{\frac{n-1}{2}})(2^{n-1})+(2^{n-1}-2^{\frac{n+1}{2}})(2^{n-1}-1))} \\&
            =&2^{n-1}-\frac{1}{2}\sqrt{2^{\frac{3n+1}{2}}+2^{\frac{3n-1}{2}}+2^n-2^{\frac{n+3}{2}}}\\&
            \ge& 2^{n-1}-2^{\frac{3n-5}{4}+\frac{1}{2}\log_2 3}-O(2^{\frac{n}{4}}). 
      \end{eqnarray*}
Similarly, when $n$ is odd and $3\mid n$, we have 
\begin{eqnarray*}
    \nl_2(f) & \ge &  2^{n-1}-\frac{1}{2}\sqrt{2^{\frac{3n+3}{2}}+2^n-2^{\frac{n+3}{2}}-\wt(\tr_n(x^7))\cdot 2^{\frac{n+1}{2}}} \\
& \ge & 2^{n-1}-\frac{1}{2}\sqrt{3\cdot 2^{\frac{3n-1}{2}}+2^n +3\cdot 2^{n+\frac{1}{2}}-2^{\frac{n+3}{2}}}
    \\ & \ge &2^{n-1}-2^{\frac{3n-5}{4}+\frac{1}{2}\log_2 3}-O(2^{\frac{n}{4}}),
\end{eqnarray*}
where the second step is because $\wt(\tr_n(x^7))\ge 2^{n-1}- 3\cdot2^{\frac{n}{2}}$ by Lemma \ref{Weil's bound}.
\end{proof}

\subsection{Functions of the type $\tr_n(x^{2^r+3})$ for $n=2r$}

In \cite{YT20}, Yan and Tang proved lower bounds on the second-order nonlinearity of the functions $\tr_n(x^{2^{r+1}+3})$, where $n = 2r$. This class of functions was first studied by Cusick and Dobbertin \cite{CD96}. We study a similar, but different, class of functions, that is, $\tr_n(x^{2^{r}+3})$ for $n = 2r$. In terms of techniques, our proof is similar to \cite{YT20}, and the lower bound is the same as that in \cite{YT20}. Our main contribution is to \emph{identify} this class of functions for the first time. 

Let $f=\tr_n(x^{2^r+3})$. By Proposition \ref{lower_non}, we can estimate the second-order nonlinearity $\nl_2(f)$ by calculating the nonlinearity of the derivatives of $f$, denoted by $D_a f$. We have
\begin{eqnarray*}
    D_af(x) & = & \tr_n(x^{2^r+3} + (x+a)^{2^r+3}) \\
    & = & \tr_n(a^{2^r}x^3+a^2x^{2^r+1}+ax^{2^r+2}) + \tr_n(a^3x^{2^r} + a^{2^r+1}x^2 + a^{2^r+2}x +  a^{2^r+3}),
\end{eqnarray*}
where $\tr_n(a^3x^{2^r} + a^{2^r+1}x^2 + a^{2^r+2}x + a^{2^r+3})$ is an affine function. 

\begin{theorem}\label{dim_E_ga}
    Let $\mathcal{E}_{D_a f}$ be the linear kernel of $D_a f(x)$. For odd $r$, we have
\begin{equation*}
\mathrm{dim}(\mathcal{E}_{D_af}) =
\begin{cases}
    r+1, & a\in \mathbb{F}_{2^r}^*, \\
    2,  & a\in \mathbb{F}_{2^n} \setminus \mathbb{F}_{2^r}.
\end{cases}
\end{equation*}
Let $G=\{g^{3s}\mid 0\le s \le \frac{2^r-4}{3}\}$ and $g$ is a primitive element of $\mathbb{F}_{2^r}$. For even $r$, we have

\begin{equation*}
\mathrm{dim}(\mathcal{E}_{D_af}) =
\begin{cases}
    r+2, & a\in G, \\
    r,  & a \in \mathbb{F}_{2^r}^* \setminus G, \\
    2,  & a \in \mathbb{F}_{2^n} \setminus \mathbb{F}_{2^r}.
\end{cases}
\end{equation*}

\end{theorem}

\begin{proof} Let $g_a(x)= \tr_n(a^{2^r}x^3+a^2x^{2^r+1}+ax^{2^r+2})$.
By the definition of the linear kernel, we have 
\begin{align*}
    \mathcal{E}_{D_a f} =\mathcal{E}_{g_a} = \{x\in \mathbb{F}_{2^n} \mid B(x,y)=g_a(0)+g_a(x)+g_a(y)+g_a(x+y) = 0,\ \mbox{for\ all}\ y \in \mathbb{F}_{2^n}\}.
\end{align*}
Using the properties of the trace function and the fact that $n = 2r$, we have
 \begin{eqnarray}\label{kernel_ga}
    0&=& B(x,y)\nonumber \\
    &=&g_a(0)+g_a(x)+g_a(y)+g_a(x+y) \nonumber \\
    &=& \tr_n(a^{2^r}(x^2y+xy^2) + a(x^{2^r}y^2+x^2y^{2^r}) + a^2(x^{2^r}y+xy^{2^r}))  \nonumber \\
    &=& \tr_n((a^{2^r}x^2+a^2x^{2^r})y+(a^{2^r}x+ax^{2^r})y^2+(ax^2+a^2x)y^{2^r}) \nonumber \\
    &=& \tr_n((a^{2^r}x^2+a^2x^{2^r}+a^{2^{r-1}}x^{2^{n-1}}+a^{2^{n-1}}x^{2^{r-1}}+a^{2^r}x^{2^{r+1}}+a^{2^{r+1}}x^{2^r})y) .
 \end{eqnarray}
Equation \eqref{kernel_ga} holds for all $y\in \mathbb{F}_{2^n}$ if and only if the coefficient of $y$ is zero, that is, 
\begin{equation}
\label{equ:thm2_ax}
    a^{2^r}x^2+a^2x^{2^r}+a^{2^{r-1}}x^{2^{n-1}}+a^{2^{n-1}}x^{2^{r-1}}+a^{2^r}x^{2^{r+1}}+a^{2^{r+1}}x^{2^r}=0.
\end{equation}
Let $
\begin{cases} 
y=x^{2^r}\\
b=a^{2^r}
\end{cases}
$. Thus
$
\begin{cases} 
x=y^{2^r}\\
a=b^{2^r}
\end{cases}
$.
Equation \eqref{equ:thm2_ax} becomes 
\begin{equation*}
    bx^2+a^2y+b^{\frac{1}{2}}x^{2^{n-1}}+a^{2^{n-1}}y^{\frac{1}{2}}+by^2+b^2y=0.
\end{equation*}
Squaring both sides of the above equation, we have
\begin{eqnarray}
   0&=& b^2x^4+a^4y^2+bx+ay+b^2y^4+b^4y^2 \nonumber\\
  \label{E1} &=& b^2(x+y)^4+y^2(a^4+b^4)+bx+ay \\
  \label{E2} &=& a^{2^{r+1}}(x^4+x^{2^{r+2}})+x^{2^{r+1}}(a^4+a^{2^{r+2}})+a^{2^r}x+ax^{2^r}.
\end{eqnarray}
Thus $\mathcal{E}_{D_af}$ is the set of $x\in \mathbb{F}_{2^n}$ such that \eqref{E1} is satisfied. We consider the following cases.

\textbf{Case 1}: $a\notin \mathbb{F}_{2^r}$, i.e., $a\neq b$.

\textbf{Subcase 1.1}: $x\in \mathbb{F}_{2^r}$, i.e. $x=y$. In this case, \eqref{E1} is equivalent to 
\begin{eqnarray}\label{subcase1_1}
   0 & = &(a^4+b^4)x^2+(a+b)x \nonumber \\
    & = & (a+b)x((a+b)^3x + 1).
\end{eqnarray}
The solutions to \eqref{subcase1_1} are $x \in \{0, (a+b)^{2^n-4}\}$.

\textbf{Subcase 1.2}: $x\notin \mathbb{F}_{2^r}$. Since $(a^{2^r}x+ax^{2^r})^{2^r} = a^{2^r}x+ax^{2^r}$, we have $a^{2^r}x+ax^{2^r}\in \mathbb{F}_{2^r}$. From \eqref{E2}, we have $$
a^{2^{r+1}}(x^4+x^{2^{r+2}})+x^{2^{r+1}}(a^4+a^{2^{r+2}})=a^{2^r}x+ax^{2^r},
$$
which implies that $a^{2^{r+1}}(x^4+x^{2^{r+2}})+x^{2^{r+1}}(a^4+a^{2^{r+2}})\in \mathbb{F}_{2^r}$. Since any element $\alpha \in \mathbb{F}_{2^r}$ satisfies equation $\alpha^{2^r} = \alpha$, we have
\begin{eqnarray}\label{x_neq_y}
    0&=& \left(a^{2^{r+1}}(x^4+x^{2^{r+2}})+x^{2^{r+1}}(a^4+a^{2^{r+2}}) \right)^{2^r} +\left(a^{2^{r+1}}(x^4+x^{2^{r+2}})+x^{2^{r+1}}(a^4+a^{2^{r+2}}) \right) \nonumber\\
    &=& (a^2+a^{2^{r+1}})(x^2+x^{2^{r+1}})^2+(x^2+x^{2^{r+1}})(a^2+a^{2^{r+1}})^2  \nonumber \\
    &=& (a^2+a^{2^{r+1}})(x^2+x^{2^{r+1}})( a^2+a^{2^{r+1}} + x^2+x^{2^{r+1}}).
\end{eqnarray}
Since $a,x\notin \mathbb{F}_{2^r}$, we have $a^2+a^{2^{r+1}} \not= 0$ and $x^2+x^{2^{r+1}} \not= 0$. Thus, by \eqref{x_neq_y}, we have $a^2+a^{2^{r+1}} + x^2+x^{2^{r+1}} = 0$, that is, $(x+a)^2=(x+a)^{2^{r+1}}$. As such, we have $x+a=(x+a)^{2^r}$, that is, $y=x+a+b$. So we claim that if $x\in \mathcal{E}_{D_af}$, then we must have $y=x+a+b$. On the other hand, let us solve  \eqref{E1} assuming $y = x+a+b$ is satisfied, which, in fact, must be satisfied, as we have shown. Plugging $y=x+a+b$ into equation \eqref{E1}, we have
\begin{eqnarray}
0&= &b^2(a+b)^4+(x+a+b)^2(a+b)^4+bx+a(x+a+b) \nonumber\\
    &= & (a+b)^4x^2+a^2(a+b)^4+(a+b)x+a(a+b) \nonumber \\
    &= & (a+b)^4(x+a)^2+(a+b)(x+a) \nonumber \\
    & = & (a+b)(x+a)((a+b)^3(x+a) + 1)\nonumber ,
\end{eqnarray}
which implies that $x=a$ or $x=(a+b)^{2^n-4} + a$.

In Case 1 where  $a\notin \mathbb{F}_{2^r}$, we conclude that $\mathcal{E}_{D_af}=\{0, (a+b)^{2^n-4},a,(a+b)^{2^n-4}+a\}$ and $\dim (\mathcal{E}_{D_af}) = 2$.

\hspace*{\fill}

\textbf{Case 2}: $a\in \mathbb{F}_{2^r}^*$. In this case, $b = a$, equation \eqref{E1} becomes
\begin{eqnarray}\label{Case_2}
    0&=&a^2(x+y)^4+a(x+y) \nonumber \\
    &= &a(x+y)(a(x+y)^3+1),
\end{eqnarray}
which implies that $y = x$ or $(x+y)^3=a^{2^r-2}$.

\textbf{Subcase 2.1}: $y=x$, i.e., $x^{2^r}=x$. In this case, $x^{2^r}=x$ if and only if $x\in \mathbb{F}_{2^r}$. Thus $\mathbb{F}_{2^r} \subseteq \mathcal{E}_{D_af}$.

\textbf{Subcase 2.2}: $(x+y)^3=a^{2^r-2}$. In this case, we consider the following two subcases according to the parity of $r$.

\begin{itemize}
    \item If $r$ is odd, we have $2^r-1 \equiv 1 \pmod 3$ and $\gcd(2^r-2,3)=3$. Thus $(x+y)^3=a^{2^r-2} $ implies 
    \begin{equation}\label{E4}
        x^{2^r}+x=a^{\frac{2^r-2}{3}}.
    \end{equation}
    Since $\mathbb{F}_{2^n}$ is a field extension of $\mathbb{F}_{2^r}$ of degree $2$, $x^{2^r}+x$ is the trace function from $\mathbb{F}_{2^n}$ to $\mathbb{F}_{2^r}$, 
    which is a $2^r$-to-1 mapping. So the number of solutions to \eqref{E4} is $2^r$. Combining with Subcase 2.1, we conclude that when $r$ is odd and $a \in \mathbb{F}_{2^r}^*$, we have $\dim (\mathcal{E}_{D_af})=r+1$.
    \item If $r$ is even, we have $\gcd(2^r-1,3)=3$. Let $G=\{g^{3s}\mid 0\le s \le \frac{2^r-4}{3}\}$ be the multiplicative group of order $\frac{2^r-1}{3}$ and $g$ is a primitive element of $\mathbb{F}_{2^r}$. If $a\notin G$, then $a^{2^r-2}$ is not a cube, that is, $(x+y)^3=a^{2^r-2}=a^{-1}$ has no roots. Combining with Subcase 2.1, we deduce that $\dim (\mathcal{E}_{g_a})=r$ when $a \notin G$ and $r$ is even. If $a\in G$, we have
\begin{equation}\label{E5}
    x+y=x^{2^r}+x=g^{-s+\frac{2^r-1}{3}i},\ \mbox{for}\ i=0,1,2,
\end{equation}
for some $0 \le s \le \frac{2^r-4}{3}$. Similarly, we can prove, for each $i = 0, 1, 2$, equation \eqref{E5} has exactly $2^r$ solutions. Thus we have $\dim (\mathcal{E}_{g_a})=r+2$ when $a\in G$ and $r$ is even. 
\end{itemize}

Summarizing all the cases above, we complete the proof.
\end{proof}

Combining Proposition \ref{lower_non} and Theorem \ref{dim_E_ga}, we can prove the the lower bound on the second-order nonlinearity of $\tr_n(x^{2^r+3})$ for $n=2r$.

\begin{proof} (of Theorem \ref{theorem_2})
When $r$ is odd, by Lemma \ref{Walsh_spec_quar} and Theorem \ref{dim_E_ga}, we have
    \begin{equation}\label{odd_r}
\nl(D_af(x)) = 
\begin{cases}
    2^{n-1}-2^{\frac{n+r-1}{2}},  &a\in \mathbb{F}_{2^r}^*,\\
    2^{n-1}-2^{\frac{n}{2}},  &a\in \mathbb{F}_{2^n} \setminus \mathbb{F}_{2^r}.
\end{cases}
\end{equation}
By Proposition \ref{lower_non} and \eqref{odd_r}, we have
\begin{eqnarray*}
    \nl_2(f) &\ge& 2^{n-1} - \frac{1}{2}\sqrt{2^{\frac{3n}{2}+1} + 2^{\frac{5n}{4}+\frac{1}{2}}- 2^n  - 2^{\frac{3n}{4}+\frac{1}{2}}} \\&
    = &2^{n-1}-2^{\frac{3n}{4}-\frac{1}{2}}-O(2^{\frac{n}{2}}).
\end{eqnarray*}
When $r$ is even, similarly, we can prove
\begin{equation} \label{even_r}
\nl(D_af(x)) = 
\begin{cases}
    2^{n-1}-2^{\frac{n+r}{2}}, &a\in G, \\
    2^{n-1}-2^{\frac{n+r}{2}-1}, &a\in \mathbb{F}_{2^r}^*\setminus G, \\
    2^{n-1}-2^{\frac{n}{2}}, & a\in \mathbb{F}_{2^n} \setminus \mathbb{F}_{2^r}.
\end{cases}
\end{equation}
where $G=\{g^{3s}\mid 0\le s \le \frac{2^r-4}{3}\}$ is the multiplicative group of order $\frac{2^r-1}{3}$ and $g$ is a primitive element of $\mathbb{F}_{2^r}$. By Proposition \ref{lower_non} and \eqref{even_r}, we have
\begin{eqnarray*}
        \nl_2(f)&\ge &2^{n-1}-\frac{1}{2}\sqrt{2^{\frac{3}{2}n+1}+\frac{1}{3}\cdot 2^{\frac{5}{4}n+2}-2^n-\frac{1}{3}\cdot 2^{\frac{3}{4}n+2}}\\&
       = &2^{n-1}-2^{\frac{3n}{4}-\frac{1}{2}}-O(2^{\frac{n}{2}}).
\end{eqnarray*}

\end{proof}

\section{Third-order nonlinearity}

The following proposition is proved by applying Proposition \ref{lower_non} twice.

\begin{proposition}\label{pro3}\cite{Car08}
    Let $f$ be any $n$-variable function and $r$ a positive integer smaller than $n$. We have
    \[
    \nl_r(f)\ge 2^{n-1}-\frac{1}{2}\sqrt{\sum_{a\in \mathbb{F}_{2^n}}\sqrt{2^{2n}-2\sum_{b\in \mathbb{F}_{2^n}} \nl_{r-2}(D_aD_bf)}}.
    \]
\end{proposition}

By the above proposition, our goal is to estimate the nonlinearities of the second-order derivatives of $\tr_n(x^{15})$. Observe that
\begin{eqnarray*}
\sum_{a\in \mathbb{F}_{2^n}}\sqrt{2^{2n}-2\sum_{b\in \mathbb{F}_{2^n}} \nl_{r-2}(D_aD_bf)}
& = & \sum_{a\in \mathbb{F}_{2^n}}\sqrt{2^{2n}-2\sum_{b\in \mathbb{F}_{2^n}} \nl_{r-2}(D_aD_{ab}f)} \\    
& = & \sum_{a\in \mathbb{F}_{2^n}}\sqrt{2^{2n}-2\sum_{b\in \mathbb{F}_{2^n}} \nl_{r-2}(D_{ab}D_{a}f)}.
\end{eqnarray*}
Thus it is equivalent to estimate the first-order nonlinearity of $D_{ab}D_{a}f$ for all $a, b \in \mathbb{F}_{2^n}$.

\begin{lemma}
\label{lem:PQR}
Let $f = \tr_n(x^{15})$. For any  $a\in \mathbb{F}_{2^n}$ and $b\in \mathbb{F}_{2^n}$, element $x\in \mathbb{F}_{2^n}$ is in the linear kernel of $D_{ab}D_af$ if and only if $P(x,a,b)=0$, where
\begin{equation}
\label{equ:def_P}
P(x, a, b)=Q(x,a,b)(Q(x,a,b)+1)
\end{equation}
and 
\begin{equation}
\label{equ:def_Q}
Q(x,a,b)=(b^2+b)^{-4} R(x,a,b)(R(x,a,b)+1)
\end{equation}
and
\begin{equation}
\label{equ:def_R}
R(x,a,b)=a^{30}(b^2+b)^6\left((x^2+x)^2+(x^2+x)(b^2+b)\right)^4+ a^{15}(b^2+b)^5\left((x^2+x)^2+(x^2+x)(b^2+b)\right).
\end{equation}
\end{lemma}
\begin{proof}
For any $a, b \in \mathbb{F}_{2^n}$, we have 
\begin{eqnarray*}
(D_af)(ax) &=& \tr_n((ax)^{15})+\tr_n((ax+a)^{15}) \\
  &=& \tr_n((ax)^{15}+(ax+a)^{15})\\
  &=& \tr_n(a^{15} (\sum_{i=0}^{14} x^i)),
\end{eqnarray*}
and
\begin{eqnarray*}
        &&D_{b}((D_af)(ax))\\&=&(D_{ab}D_af)(ax)\\
        &=& \tr_n(a^{15}( (b^2+b)x^{12}+(b^4+b)x^{10}+(b^4+b^2)x^{9}+ (b^8+b)x^6+(b^8+b^2)x^5 +(b^8+b^4)x^3 ))+ l(x),
\end{eqnarray*}
where $l(x)$ is an affine function. By Proposition \ref{equivalent}, we have $\mathcal{E}_{D_{ab}D_af(ax)}=\mathcal{E}_{D_{ab}D_af}$ for any $b\in \mathbb{F}_{2^n} \setminus \{0,1\}$ and $a\in \mathbb{F}_{2^n}^*$. (When $a = 0$ or $b \in \{0, 1\}$, $D_{ab} D_a f$ becomes 0, so the conclusion holds obviously.)

For convenience, let $g(x)=D_{ab}D_af(ax)$. We have  $\mathcal{E}_{D_{ab}D_af(ax)}=\{x\in\mathbb{F}_{2}^n\mid B(x,y)=g(0)+g(x)+g(y)+g(x+y)=0\ \text{for all $y\in \mathbb{F}_{2}^n$}\}$ by definition. By somewhat tedious computation, we have
\begin{eqnarray*}
         B(x,y)&=& g(0)+g(x)+g(y)+g(x+y)
         \\&= &\tr_n(a^{15}((b^2+b)x^4+(b^4+b)x^2+(b^4+b^2)x)y^8)\\
         & &+\tr_n(a^{15}((b^2+b)x^8+(b^8+b)x^2+(b^8+b^2)x)y^4 )\\
         & & +\tr_n(a^{15}((b^4+b)x^8+(b^8+b)x^4+(b^8+b^4)x)y^2 )\\
         & & +\tr_n(a^{15}((b^4+b^2)x^8+(b^8+b^2)x^4+(b^8+b^4)x^2)y) .
 \end{eqnarray*}
Using the properties of the trace function, we have
\begin{eqnarray*}
         B(x,y)
          &=&  \tr_n(((a^{15}((b^2+b)x^4+(b^4+b)x^2+(b^4+b^2)x))^{2^{-3}} \\
         & & +(a^{15}((b^2+b)x^8+(b^8+b)x^2+(b^8+b^2)x))^{2^{-2}}  \\
         & & +(a^{15}((b^4+b)x^8+(b^8+b)x^4+(b^8+b^4)x))^{2^{-1}} \\ 
         & & +a^{15}((b^4+b^2)x^8+(b^8+b^2)x^4+(b^8+b^4)x^2))y) .
    \end{eqnarray*}
It is clear that $B(x,y)=0$ for all $y\in \mathbb{F}_{2^n}$ if and only if the coefficient of $y$ is zero, that is,
\begin{eqnarray*}
        0&=&(a^{15}((b^2+b)x^4+(b^4+b)x^2+(b^4+b^2)x))^{2^{-3}}
         \\ 
         &+&(a^{15}((b^2+b)x^8+(b^8+b)x^2+(b^8+b^2)x))^{2^{-2}} 
         \\
         &+&(a^{15}((b^4+b)x^8+(b^8+b)x^4+(b^8+b^4)x))^{2^{-1}} 
         \\
         &+&a^{15}((b^4+b^2)x^8+(b^8+b^2)x^4+(b^8+b^4)x^2).
\end{eqnarray*}
Raising both sides of the above equation to the $8$th power, we get $P(x, a, b) = 0$, as desired.
\end{proof}

Let $\numroots_P(a,b)$ denote by the number of $x\in \mathbb{F}_{2^n}$ such that $P(x,a,b)=0$ where $a\neq 0$ and $b\neq 0,1$; let $\numroots_Q(a,b)$ denote by the number of $x\in \mathbb{F}_{2^n}$ such that $Q(x,a,b)=0$; let $\numroots_{Q+1}(a,b)$ denote by the number of $x\in \mathbb{F}_{2^n}$ such that $Q(x,a,b)+1=0$; let $\numroots_R(a,b)$ denote the number of $x\in \mathbb{F}_{2^n}$ such that $R(x,a,b)=0$, and $\numroots_{R+1}(a,b)$ denote the number of $x\in \mathbb{F}_{2^n}$ such that $R(x,a,b)+1=0$.

\begin{lemma}
 \label{lem:N_PQR_range}
 Let $a \in \mathbb{F}_{2^n}^*$ and $b \in \mathbb{F}_{2^n} \setminus \{0, 1\}$, and let polynomials $P(x, a, b)$, $Q(x, a, b)$ and $R(x, a, b)$ be defined as in Lemma \ref{lem:PQR}. We have $\numroots_P(a,b)=\numroots_Q(a,b)+\numroots_{Q+1}(a,b)$, and $ \numroots_Q(a,b)=\numroots_R(a,b)+\numroots_{R+1}(a,b)$. In addition,
 \begin{itemize}
     \item $\numroots_Q(a,b)\in \{2,2^2,2^3,2^4,2^5\}$ and $\numroots_{Q+1}(a,b)\le 32$.
     \item $\numroots_R(a,b)\in \{2^2,2^3,2^4\}$ and $\numroots_{R+1}(a,b)\le 16$.
     \item When $n$ is even, $\numroots_P(a,b) \in \{2^2,2^4,2^6\}$; when $n$ is odd, $\numroots_P(a,b) \in \{2^1,2^3,2^5\}$.
\end{itemize}
\end{lemma}
\begin{proof}

  Notice that $P(x,a,b)$ is a $2$-polynomial (in variable $x$) of degree 64; the number of roots for equation $P(x,a,b)=0$ is at most 64. By Lemma \ref{parity_dim}, we know that $\dim (\mathcal{E}_{D_{ab}D_af})$ and $n$ have the same parity. Therefore, when $n$ is even, $\numroots_P(a,b)\in \{2^2,2^4,2^6\}$; when $n$ is odd, $\numroots_P(a,b)\in \{2^1,2^3,2^5\}$. (Note that, when $b \notin \{0, 1\}$, $R(x, a, b) = 0$ has at least 4 roots $0, 1, b, b+1$, which implies that $P(x, a, b) = 0$ has at least 4 roots.)

 Since $P(x, a, b) = Q(x, a, b)(Q(x, a, b) + 1)$, we have $\numroots_P(a,b)=\numroots_Q(a,b)+\numroots_{Q+1}(a,b)$. Observe that $Q(x,a,b)$ is a $2$-polynomial of degree 32, we have $\numroots_Q(a,b)\in \{2,2^2,2^3,2^4,2^5\}$.

  From \eqref{equ:def_Q}, we have $\numroots_Q(a,b)=\numroots_R(a,b)+\numroots_{R+1}(a,b)$, when $b \notin \{0, 1\}$. Clearly, $R(x,a,b)$ is a $2$-polynomial of degree 16 in variable $x$. Note that $x=0,1,b,b+1$ are the four different roots whenever $b\in \mathbb{F}_{2^n}\setminus \{0,1\}$, then $\numroots_R(a,b)\in \{2^2,2^3,2^4\}$. On the other hand, the degree of $R(x,a,b)+1$ is 16, so $\numroots_{R+1}(a,b)\le 16$. Since $Q(x,a,b)$ is a 2-polynomial of degree 32, we have $\numroots_{Q}(a,b)\in \{2^2,2^3,2^4,2^5\}$. In the following, we lower bound  the number of $b\in \mathbb{F}_{2^n}\setminus \{0,1\}$ such that $\numroots_{P}(a,b)\le 16$ for a fixed $a\in \mathbb{F}_{2^n}^*$. 
\end{proof}

\begin{lemma}\label{dim_4_even}
Let $n$ be even. Let $a \in \mathbb{F}_{2^n}^*$ and $b \in \mathbb{F}_{2^n} \setminus \{0, 1\}$. If $\numroots_R(a,b)= 4$, then $\numroots_P(a,b)\le 16$. 
\end{lemma}
\begin{proof}
Since $\numroots_R(a, b) = 4$, we have $\numroots_Q(a, b) \le 4 + \deg(R+1) = 20$. Note that $\numroots_{Q}(a,b)=\numroots_{R}(a,b) + \numroots_{R+1}(a,b) \in \{2^2,2^3,2^4,2^5\}$ by Lemma \ref{lem:N_PQR_range}. So we have $\numroots_Q(a, b) \le 16$.

Note that $\numroots_{P}(a,b)=\numroots_{Q}(a,b)+\numroots_{Q+1}(a,b)$, and $\numroots_{P}(a,b) \in \{2^2,2^4,2^6\}$ by Lemma \ref{lem:N_PQR_range}. So we have $\numroots_{P}(a,b) \le 16$.
\end{proof}

By Lemma \ref{dim_4_even}, when $n$ is even, to lower bound the number of $b\in \mathbb{F}_{2^n}\setminus \{0,1\}$ where $ \dim(\mathcal{E}_{D_{ab}D_af}) \le 4$, it suffices to lower bound the number of $b\in \mathbb{F}_{2^n}\setminus \{0,1\}$ where $ \numroots_R(a,b)=4$.

\begin{theorem}\label{th:N_R_even}
Let $n$ be even. For any $a \in \mathbb{F}_{2^n}^*$, there are at least $\frac{1}{3}\cdot (2^{n+1}-2^{\frac{n}{2}+1}-4)$ elements $b \in \mathbb{F}_{2^n}\setminus \{0,1\} $ such that $\numroots_R(a,b) = 4$.
\end{theorem}

\begin{proof}
When $x\notin \{0,1,b,b+1\}$, we have $(x^2+x)^2+(x^2+x)(b^2+b)\neq 0$. Let $y=(x^2+x)^2+(x^2+x)(b^2+b)$. Since $R(x,a,b)=0$, we can deduce that
\begin{equation}
\label{equ:ycubeab}
    y^3=\frac{1}{a^{15}(b^2+b)}.
\end{equation}

Let $G=\{g^{3s}\mid 0\le s\le \frac{2^n-4}{3}\}$, where $g$ is a primitive element. If $b^2+b \not\in G$, it is clear that \eqref{equ:ycubeab} has no solution, which implies that $\solnum_R( a, b) = 4$. Next, we prove there are at least $\frac{1}{3} \cdot(2^{n+1}-2^{\frac{n}{2}+1}-4)$ elements $b$ such that $b^2+b \not\in G$, which will complete our proof.

We estimate the number of elements $b\in \mathbb{F}_{2^n} \setminus \{0,1\}$ such that $b^2+b\notin G$. Let $s_1$ denote the number of elements $b\in \mathbb{F}_{2^n} \setminus \{0,1\}$ such that $b^2+b\notin G$; let $s_2$ denote the number of elements $b\in \mathbb{F}_{2^n} \setminus \{0,1\}$ such that $b^2+b\in G$. 

Consider equation 
\begin{equation}
\label{equ:xcube_bsquare_b}
x^3= b^2+b.
\end{equation}
Observe that
\begin{itemize}
    \item Equation \eqref{equ:xcube_bsquare_b} has (at least) a solution in variable $x$ if and only if $b^2+b \in G$.
    \item It is well known that (for example, see page 536 in \cite{Car21}) equation $b^2+b=c$ has two solutions if and only if $\tr_n(c) = 0$, otherwise the equation has no solution. Thus, equation \eqref{equ:xcube_bsquare_b} in variable $b$ has a solution if and only if $\tr_n(x^3) = 0$.
\end{itemize}
 For any fixed $b$, denote the set of solutions by $X_b$, and let $X = \cup_{b \in \mathbb{F}_{2^n}\setminus \{0,1\}} X_b$. Consider the mapping $\phi(x) : X \to \mathbb{F}^*_{2^n}$, where $\phi(x) = x^3$. Notice that $\phi(x): \mathbb{F}_{2^n}^* \to \mathbb{F}_{2^n}^*$ is a 3-to-1 mapping on $\mathbb{F}_{2^n}^*$. Furthermore, mapping $\phi: X \to \mathbb{F}^*_{2^n}$ is also a 3-to-1 mapping on $X$. Otherwise, there exist $x_1\in X$, $x_2\notin X$ such that $x_1^3=x_2^3$ and $\tr_n(x_1^3)\neq \tr_n(x_2^3)$, which is a contradiction.

Recall that there exists $b$ such that $x^3 = b^2 + b$ if and only if $\tr_n(x^3) = 0$. Therefore, $|X| = 2^n - 1 - \wt(\tr_n(x^3))$. Combining with the fact that $\phi(x): X \to \mathbb{F}_{2^n}^*$ is a 3-to-1 mapping, we have
\[
|\{b^2 + b : b^2 + b \in G, b \in \mathbb{F}_{2^n}\setminus \{0,1\}| = \frac{1}{3} \cdot (2^n - 1 - \wt(\tr_n(x^3))).
\]
Since $b\mapsto b^2+b$ is 2-to-1 mapping on $\mathbb{F}_{2^n}\setminus \{0,1\}$, we have $s_2= 2\cdot |\{b^2 + b : b^2 + b \in G\}|=\frac{1}{3} \cdot (2^{n+1}-2-2\wt(\tr_n(x^3)))$. So $s_1= 2^n-2-s_2=\frac{1}{3} \cdot (2^n-4+2\wt(\tr_n(x^3)))$. By Lemma \ref{Weil's bound}, we have $\wt(\tr_n(x^3))\ge 2^{n-1}-2^{\frac{n}{2}}$. So $s_1 \ge \frac{1}{3} \cdot (2^{n+1}-2^{\frac{n}{2}+1}-4)$. 

The number of $b\in \mathbb{F}_{2^n}\setminus \{0,1\}$ such that \eqref{equ:ycubeab} with no solution is at least $\frac{1}{3}\cdot (2^{n+1}-2^{\frac{n}{2}+1}-4)$. 
\end{proof}

By Theorem \ref{th:N_R_even}, the following theorem is immediate.

\begin{theorem}\label{dim4_even}
    Let $n$ be even. We have
    \begin{table}[H]
	\begin{center}
 \caption{The distribution of $\dim(\mathcal{E}_{D_{ab}D_af})$ for any fixed $a\in \mathbb{F}_{2^n}^*$, even $n$} 
        \renewcommand\arraystretch{1.5}
		\begin{tabular}{ c|c|c|c } 
			\toprule 
			$\dim(\mathcal{E}_{D_{ab}D_af}) $ & n& $\{2, 4\}$ & $6 $ \\ 
			\hline
			The number of $b\in \mathbb{F}_{2^n}$ & 2& $\ge \frac{2^{n+1}-2^{\frac{n}{2}+1}-4}{3}$  &  $\le \frac{2^n+2^{\frac{n}{2}+1}-2}{3}$   \\ 
			\bottomrule
		\end{tabular}
		
	\end{center}
\end{table}
\end{theorem}

\begin{lemma}\label{dim4_odd_prev} Let $n$ be odd.
Let $a \in \mathbb{F}_{2^n}^*$ and $b \in \mathbb{F}_{2^n}\setminus \{0,1\}$.
If $ \numroots_R(a,b)= 4$, then $\numroots_P(a,b)\le 8$. 
\end{lemma}
\begin{proof}
Let $n$ be odd and let $\numroots_R(a, b) = 4$. We will prove the followings step by step:
\begin{itemize}
    \item $\numroots_{R+1}(a, b) \le 8$.
    \item $\numroots_Q(a, b) \le 8$.
    \item $\numroots_{Q+1}(a, b) \le 16$.
    \item $\numroots_P(a, b) \le 8$.
\end{itemize}

First, let us prove $\numroots_{R+1}(a, b) \le 8$.
If $R(x,a,b)+1=0$, we have
\begin{eqnarray}
   0&=& a^{30}(b^2+b)^6\big( (x^2+x)^2+(b^2+b)(x^2+x)\big)^4+a^{15}(b^2+b)^5\big((x^2+x)^2+(b^2+b)(x^2+x)\big)+1 \label{y_5}\\
    &=&a^{30}(b^2+b)^6y^4+a^{15}(b^2+b)^5y+1,\label{y6}
\end{eqnarray}
where $y=(x^2+x)^2+(b^2+b)(x^2+x)$. Equation \eqref{y6} can be converted to 
\begin{equation}\label{y5}
    y^4+\frac{y}{a^{15}(b^2+b)}+\frac{1}{a^{30}(b^2+b)^6}=0.
\end{equation}
By Lemma \ref{num_root_quar}, since $n$ is odd, equation \eqref{y5} in variable $y$ has no solution or exactly two solutions. Furthermore, since $(x^2+x)^2+(b^2+b)(x^2+x)$ is a polynomial of degree 4, the number of $x\in \mathbb{F}_{2^n}$ such that $(x^2+x)^2+(b^2+b)(x^2+x)+y=c$, for any $c$, is at most 4. So equation \eqref{y_5} in variable $x$ has at most 8 solutions, that is, $\numroots_{R+1}(a,b)\le 8$.

Second, we prove $\numroots_Q(a, b) \le 8$.
Note that $\numroots_{Q}(a,b)=\numroots_{R}(a,b)+\numroots_{R+1}(a,b)\le 12$. On the other hand, by Lemma \ref{lem:N_PQR_range}, $\numroots_{Q}(a,b) \in \{2, 2^2, 2^3, 2^4, 2^5\}$. So we have $ \numroots_{Q}(a,b) \le 8$. 

Next, we prove $\numroots_{Q+1}(a, b) \le 16$. Suppose $Q(x,a,b)+1=0$. We have $(b^2+b)^4Q(x,a,b)+(b^2+b)^4=0$, that is,
\begin{equation}
\label{equ:step3_quadR}
R(x,a,b)(R(x,a,b)+1) + (b^2+b)^4=0.    
\end{equation}
Viewing \eqref{equ:step3_quadR} as a quadratic equation in variable $R$, we know that \eqref{equ:step3_quadR} has at most 2 solutions, denoted by $c_1, c_2$. We shall prove that, for each $i = 1, 2$, $R(x, a, b) = c_i$ has at most 8 solutions.

Let $R(x,a,b)=c_i$ and let $y=(x^2+x)^2+(b^2+b)(x^2+x)$, where $c_i\in \mathbb{F}_{2^n}^*$. Then we have 
\[
a^{30}(b^2+b)^6y^4+a^{15}(b^2+b)^5y+c_i=0,
\]
that is
\begin{equation}\label{odd_eq}
    y^4+ \frac{y}{a^{15}(b^2+b)}+\frac{c_i}{a^{30}(b^2+b)^6}=0.
\end{equation}
By Lemma \ref{num_root_quar}, equation \eqref{odd_eq}, in variable $y$, has no solution or exactly two solutions. Furthermore, since $(x^2+x)^2+(b^2+b)(x^2+x)$ is a polynomial of degree 4, the number of $x\in \mathbb{F}_{2^n}$ such that $(x^2+x)^2+(b^2+b)(x^2+x)=d$ is at most 4. In total, equation $R(x, a, b) = c_i$ has at most 8 solutions.
Thus, the $Q(x, a, b) + 1 = 0$ has at most $16$ solutions, that is,
$\numroots_{Q+1}(a,b)\le 16$. 

Finally, we prove $\numroots_P(a, b) \le 8$. Note that $\numroots_{P}(a,b)=\numroots_{Q}(a,b)+\numroots_{Q+1}(a,b)\le 24$. By Lemma \ref{lem:N_PQR_range}, $\numroots_{P}(a,b)\in \{1,8,32\}$. So we have $\numroots_{P}(a,b)\le 8$.
\end{proof}

\begin{theorem}\label{dim4_odd}
    Let $n$ be odd. We have
\begin{table}[H]
	\begin{center}
 \caption{The distribution of $\dim(\mathcal{E}_{D_{ab}D_af})$ for any fixed $a\in \mathbb{F}_{2^n}^*$, odd $n$} 
        \renewcommand\arraystretch{1.5}
		\begin{tabular}{ c|c|c|c } 
			\toprule 
			$\dim(\mathcal{E}_{D_{ab}D_af}) $ & n & $  \{ 1, 3\} $ & $ 5 $ \\ 
			\hline
			The number of $b\in \mathbb{F}_{2^n}$ & 2& $\ge 3\cdot 2^{n-4}+10$  &  $\le 13\cdot 2^{n-4}-12  $   \\ 
			\bottomrule
		\end{tabular}
		
	\end{center}
\end{table}
\end{theorem}

\begin{proof}
Since $n$ is odd, we have $3\mid (2^n-2)$. When $x\in \{0,1,b,b+1\}$, we have $(x^2+x)^2+(b^2+b)(x^2+x)\neq 0 $. If $R(x,a,b)=0$ for $a\in \mathbb{F}_{2^n}^*$ and $b\in \mathbb{F}_{2^n}\setminus \{0,1\}$, we have
\begin{equation*}
    a^{15}(b^2+b)\big((x^2+x)^2+(b^2+b)(x^2+x) \big)^3=1,
\end{equation*}
which is 
\begin{equation}\label{cpy3}
(x^2+x)^2+(b^2+b)(x^2+x)=a^{-5}(b^2+b)^{\frac{2^n-2}{3}}. 
\end{equation}
Multiplying $\frac{1}{(b^2+b)^2}$ to both sides of \eqref{cpy3}, we get
\begin{equation}\label{odd_N_1}
  \left(\frac{x^2+x}{b^2+b}\right)^2+\frac{x^2+x}{b^2+b}=a^{-5}(b^2+b)^{\frac{2^n-2}{3}-2}.  
\end{equation}
If $\tr_n(a^{-5}(b^2+b)^{\frac{2^n-2}{3}-2})=1$, then $t^2+t=a^{-5}(b^2+b)^{\frac{2^n-2}{3}-2}$ has no solution, where $t=\frac{x^2+x}{b^2+b}$. So $\numroots_{R}(a,b)=4$. Thus it suffices to lower bound the number of elements $b\in \mathbb{F}_{2^n}\setminus \{0,1\}$ such that $\tr_n(a^{-5}(b^2+b)^{\frac{2^n-2}{3}-2})=1$.

Let $n=2r+1$. Note that $\frac{2^n-2}{3}-2=\sum_{i=1}^{r-1} 2^{n-2i}$.
So we have 
 \begin{eqnarray*}
        (b^2+b)^{\frac{2^n-2}{3}-2}&=&(b^2+b)^{\sum_{i=1}^{r-1}2^{n-2i}}
        \\&=& \prod_{i=1}^{r-1} (b^2+b)^{2^{n-2i}}
        \\&=& \prod_{i=1}^{r-1} (b^{2^{n-2i+1}}+b^{2^{n-2i}})
        \\&=& \sum\limits_{d_1, d_2, \ldots, d_{r-1}\in \{0,1\}} b^{\sum_{i=1}^{r-1} 2^{n-2i+d_i}}.
  \end{eqnarray*}
Expanding the trace function using its definition, we have
 \begin{eqnarray}\label{ Expanded_function}
        & &\tr_n(a^{-5}(b^2+b)^{\frac{2^n-2}{3}-2})\nonumber
        \\& = &
        \tr_n(a^{-5} \sum\limits_{d_1, \ldots, d_{r-1}\in \{0,1\}} b^{\sum_{i=1}^{r-1} 2^{n-2i+d_i}} )
        \nonumber 
        \\&=& \sum_{d_1, \ldots, d_{r-1}\in \{0,1\}} \tr_n(a^{-5}b^{\sum_{i=1}^{r-1} 2^{n-2i+d_i}} )
        \nonumber 
        \\&= & \sum_{d_1, \ldots, d_{r-1}\in \{0,1\}} \sum_{j=0}^{n-1} (a^{-5})^{2^j} b^{\sum_{i=1}^{r-1} 2^{n-2i+d_i+j}}
        \nonumber
        \\&= & \sum_{d_1, \ldots, d_{r-1}\in \{0,1\}} \sum_{j=0}^{n-1} (a^{-5\cdot 2^j}) b^{\sum_{i=1}^{r-1}2^{n-2i+d_i+j}}.
  \end{eqnarray}

 For convenience, let 
 \begin{equation}
 h(b)=\sum_{d_1, \ldots, d_{r-1}\in \{0,1\}} \sum_{j=0}^{n-1} (a^{-5\cdot 2^j}) b^{\sum_{i=1}^{r-1}2^{n-2i+d_i+j}}.
\end{equation}
Next, we will analyze the highest and lowest degree terms of the polynomial $h(b)$ as they are closely related to the number of roots of $h(b)=0$.
 
\begin{lemma}\label{number_roots}
The maximum degree of $h(b)$ is $\frac{5}{3} \cdot 2^{n-1} - \frac{32}{3}$ for $n\ge 6$. 
\end{lemma}

\begin{lemma} \label{minimum_degree}
    The minimum degree of the monomial of $h(b)$ is $\frac{1}{3}\cdot (2^{n-4}+1)$, which implies $h(b)=b^{\frac{1}{3}\cdot (2^{n-4}+1)}p(b)$, where $b\nmid p(b)$.
\end{lemma}

The proofs of these two lemmas can be found in Appendix \ref{app:proof_of_lemma_number_roots} and \ref{app:proof_of_lem_min_degree}.

By Lemma \ref{number_roots} and \ref{minimum_degree}, the number of elements $b\in \mathbb{F}_{2^n}\setminus \{0,1\}$ for which $h(b)=0$ is at most $13\cdot 2^{n-4}-12$. Hence, the number of $b\in \mathbb{F}_{2^n}\setminus \{0,1\}$ such that $h(b)=1$ is at least $3\cdot 2^{n-4}+10$.

Hence, we have the number of $b\in \mathbb{F}_{2^n} \setminus \{0,1\}$ such that $\numroots_{P}(a,b)=8$ is at least $3\cdot 2^{n-4}+10$ since
the set of $b\in \mathbb{F}_{2^n}\setminus \{0,1\}$ such that $\tr_n(a^{-5}(b^2+b)^{\frac{2^n-2}{3}-2})=1 $ satisfying is the set of roots of the equation $h(b)=1$. That is, the number of $b\in \mathbb{F}_{2^n}\setminus \{0,1\}$ such that $\dim(\mathcal{E}_{D_{ab}D_af})\le 3$ satisfying is at least $3\cdot 2^{n-4}+10$; the number of $b\in \mathbb{F}_{2^n}\setminus \{0,1\}$ such that $\dim(\mathcal{E}_{D_{ab}D_af})= 5$ is at most $13\cdot 2^{n-4}-12$. 

\end{proof}

By Theorem \ref{dim4_even} and \ref{dim4_odd}, the following corollary is immediate.
\begin{corollary}\label{cor_tr15}
    Let $f=\tr_n(x^{15})$. Denote that $\nl(D_{ab}D_af)$ be the nonlinearity of $D_{ab}D_af$. For any $b\in \mathbb{F}_{2^n}$, the distribution of $\nl(D_{ab}D_af)$ is as follows:

 \begin{table}[H]
\begin{center}
\caption{The distribution of $\nl(D_{ab}D_af)$} 
\renewcommand\arraystretch{2}
\begin{tabular}{ c|c|c } 
\hline 
n & $\nl(D_{ab}D_af)$ & The number of $b\in \mathbb{F}_{2^n}$ \\ 
\hline
\multirow{3}*{even $n$}  & $0$ & $2$ \\
                \cline{2-3}
                ~ &  $\ge 2^{n-1}- 2^{\frac{n}{2}+1}$ & $\ge \frac{1}{3} \cdot (2^{n+1}-2^{\frac{n}{2}+1}-4)$    \\ 
                \cline{2-3}
                ~ &  $\le 2^{n-1}- 2^{\frac{n}{2}+2}$ & $\le \frac{1}{3} \cdot (2^n+2^{\frac{n}{2}+1}-2)$ \\
                \hline
\multirow{3}*{odd $n$}  & $0$ & $2$ \\
                \cline{2-3}
                ~ & $\ge 2^{n-1}- 2^{\frac{n+1}{2}}$ & $\ge 3\cdot 2^{n-4}+10$   \\ 
                \cline{2-3}
                ~ & $\le 2^{n-1}- 2^{\frac{n+3}{2}}$ & $\le 13\cdot 2^{n-4}-12  $  \\
			\hline
		\end{tabular}
		
	\end{center}
\end{table}

\end{corollary}

Now we are ready to prove Theorem \ref{theorem_3}, which gives a lower bound on the third-order nonlinearity of $\tr_n(x^{15})$.  
    
\begin{proof}(of Theorem \ref{theorem_3})
By Proposition \ref{pro3} and Corollary \ref{cor_tr15}, for even $n$, we have 
 \begin{eqnarray*}
  & & \nl_3(f)  
    \\& \ge&2^{n-1}-\frac{1}{2} \sqrt{(2^n-1)\sqrt{2^{2n}-2((2^{n-1}-2^{\frac{n}{2}+1})(\frac{2^{n+1}-2^{\frac{n}{2}+1}-4}{3})+ (2^{n-1}-2^{\frac{n}{2}+2})(\frac{2^n+2^{\frac{n}{2}+1}-2}{3}))}+2^n } 
    \\& =& 2^{n-1}-\frac{1}{2} \sqrt{(2^n-1)\sqrt{\frac{1}{3}\cdot 2^{\frac{3}{2}n+4}+\frac{7}{3}\cdot 2^{n+1}-\frac{1}{3}\cdot 2^{\frac{n}{2}+5}}+2^n} \\
    &\ge & 2^{n-1}- 2^{\frac{7n}{8}-\frac{1}{4}\log_2 3}-O(2^{\frac{3n}{8}}).
 \end{eqnarray*}

By Proposition \ref{pro3} and Corollary \ref{cor_tr15}, when $n$ is odd and $n > 6$, we have 
 \begin{eqnarray*}
   & &\nl_3(f) \ge 2^{n-1}-
    \\& &\ \frac{1}{2} \sqrt{(2^n-1)\sqrt{2^{2n}-2((2^{n-1}-2^{\frac{n+1}{2}})(3\cdot 2^{n-4}+10)+(2^{n-1}-2^{\frac{n+3}{2}})(13\cdot 2^{n-4}-12))}+2^n}
    \\& =& 2^{n-1}-\frac{1}{2} \sqrt{(2^n-1)\sqrt{\frac{29}{8}\cdot 2^{\frac{3n+1}{2}}+2^{n+1}-7\cdot 2^{\frac{n+5}{2}}}+2^n} \\
    \\& \ge & 2^{n-1}- 2^{\frac{7n}{8}-\frac{13}{8}+\frac{1}{4} \log_2 29}-O(2^{\frac{3n}{8}}).
 \end{eqnarray*}
\end{proof}

\subsection{Comparison}

We list the lower bound values on the third-order nonlinearity of $\tr_n(x^{15})$ for $7\le n\le 20$ in Table \ref{table_for_odd} and \ref{table_for_even}. Our lower bound outperforms all the existing lower bounds \cite{Car08, GG10, Sin14}, both asymptotically and for all concrete $n$. 
  \begin{table}[H]
	\begin{center}
		\begin{tabular}{ c|c|c|c|c|c|c|c } 
			\hline
			$n$ & 7& 9 & 11& 13& 15& 17 &19 \\ 
			\hline
			$\nl_3$ & 12&80  & 429 & 2096  &9660  & 42923 & 186092  \\ 
			\hline
		\end{tabular}
   \caption{Lower bounds in Theorem \ref{theorem_3} for odd $n$} 
        \label{table_for_odd}
	\end{center}
\end{table}
\begin{table}[H]
	\begin{center}
		\begin{tabular}{ c|c|c|c|c|c|c|c } 
			\hline 
			$n$ & 8& 10 & 12& 14& 16& 18 &20 \\ 
			\hline
			$\nl_3$ & 30&  183 &944  & 4484 &20308 &89180  & 383411  \\ 
			\hline
		\end{tabular}
 \caption{Lower bounds in Theorem \ref{theorem_3} for even $n$}
        \label{table_for_even}
	\end{center}
\end{table}

\section{Higher-order nonlinearity}
In this section, we lower bound the $r$-th order nonlinearity for Boolean functions $\tr_n(x^{2^{r+1}-1})$ and $\tr_n(x^{2^n-2})$.

Applying $t$ times Proposition \ref{lower_non}, we have

\begin{proposition} \label{high_nonlinearity}\cite{Car08} Let $f$ be any $n$-variable Boolean function and $r$ a positive integer smaller than $n$. We have
\[
\nl_r(f)\ge 2^{n-1}- \frac{1}{2}\sqrt{ \sum_{a_1\in \mathbb{F}_{2^n}} \sqrt{ \sum_{a_2\in \mathbb{F}_{2^n}} \ldots \sqrt{ 2^{2n}-2\sum_{a_t\in \mathbb{F}_{2^n}} \nl_{r-t}(D_{a_t}D_{a_{t-1}}\ldots D_{a_1}f )} }}.
\]
    
\end{proposition}

By Proposition \ref{high_nonlinearity}, to lower bound the $r$-th order nonlinearity for functions $\tr_n(x^{2^{r+1}-1})$, our strategy is to lower bound the first-order nonlinearity $\nl(D_{a_{r-1}}D_{a_{r-2}}\ldots D_{a_1}f)$ for all distinct $a_1,a_2,\ldots,a_{r-1}\in \mathbb{F}^*_{2^n}$. We will need the following lemma in the proof of Lemma \ref{nonliearity_derivative}.

The following lemma is proved in \cite{GK12}; we state a special case of interest using different notations. Let $a, b$ be two positive integers, where $a = \sum_{i \ge 0} 2^i a_i$ and $b = \sum_{i \ge 0} 2^i b_i$ be the binary representations of $a$ and $b$ respectively. Define a partial order $\preceq$ between two positive integers as follows: $a \preceq b$ if and only if $a_i \le b_i$ for all $i \ge 0$; $a \prec b$ if and only if $a \preceq b$ and $a \neq b$. Lucas's theorem says that ${\binom{b}{a} } \equiv 1 \pmod 2$ if and only if $a \preceq b$. 

\begin{lemma} (Lemma 4 in \cite{GK12}) \label{def_rth_der}
Let $f=\tr_n(x^{2^{r+1}-1})$. For any distinct $a_1, a_2, \ldots, a_t \in \mathbb{F}_{2^n}^*$, where $1 \le t \le r$, we have
\begin{equation}
\label{der_multi_f}
D_{a_{t}}D_{a_{t-1}}\ldots D_{a_{1}}f(x) = \tr_n\Big(\sum_{ \substack{0 \prec d_{t} \prec d_{t-1} \prec \ldots \prec d_1 \prec d_0 = 2^{r+1}-1 \\  \wt(d_k)=r+1-k,\ k=1,2,\ldots,t}} x^{d_t}\prod_{i=1}^t a_i^{d_{i-1} - d_{i}}\Big) + p(x),
\end{equation}
where $\deg(p) \le r-t$.
\end{lemma}

The next lemma gives a lower bound on the first-order nonlinearity for the $(r-1)$-th order derivatives of $\tr_n(x^{2^{r+1}-1})$.

\begin{lemma} \label{nonliearity_derivative}
Let $f=\tr_n(x^{2^{r+1}-1})$. For any distinct $a_1,a_2,\ldots,a_{r-1}\in \mathbb{F}_{2^n}^*$, we have
\begin{equation*}
    \nl(D_{a_{r-1}}D_{a_{r-2}}\ldots D_{a_{1}}f)\ge 2^{n-1}-2^{\frac{n+2r-2}{2}}.
\end{equation*}
\end{lemma}
\begin{proof}
Let $g(x) = D_{a_{r-1}}D_{a_{r-2}}\ldots D_{a_{1}}f(x)$. Applying Lemma \ref{def_rth_der} with $t = r-1$, we have
\[
g(x) = \tr_n\Big(\sum_{\substack{0 \prec d_{r-1} \prec d_{r-2} \prec \ldots \prec d_1 \prec d_0 = 2^{r+1}-1 \\  \wt(d_k)=r+1-k,\ k=1,2,\ldots,r-1}} x^{d_{r-1}}\prod_{i=1}^{r-1} a_i^{d_{i-1} - d_{i}}\Big) + p(x),
\]
where $\deg(p) \le 1$.

Let $B(x, y) = g(0) + g(x) + g(y) + g(x+y)$. We have
\begin{eqnarray}
B(x, y) &=& \tr_n\Big(\sum_{\substack{0 \prec d_{r-1} \prec d_{r-2} \prec \ldots \prec d_1 \prec d_0 = 2^{r+1}-1 \\  \wt(d_k)=r+1-k,\ k=1,2,\ldots,r-1}} x^{d_{r-1}}\prod_{i=1}^{r-1} a_i^{d_{i-1} - d_{i}}\Big) + \nonumber \\
& & \tr_n\Big(\sum_{\substack{0 \prec d_{r-1} \prec d_{r-2} \prec \ldots \prec d_1 \prec d_0 = 2^{r+1}-1 \\  \wt(d_k)=r+1-k,\ k=1,2,\ldots,r-1}} y^{d_{r-1}}\prod_{i=1}^{r-1} a_i^{d_{i-1} - d_{i}}\Big) +  \nonumber \\
& & \tr_n\Big(\sum_{\substack{0 \prec d_{r-1} \prec d_{r-2} \prec \ldots \prec d_1 \prec d_0 = 2^{r+1}-1 \\ \wt(d_k)=r+1-k,\ k=1,2,\ldots,r-1}} (x+y)^{d_{r-1}}\prod_{i=1}^{r-1} a_i^{d_{i-1} - d_{i}}\Big) \nonumber \\
& = & \tr_n\Big(\sum_{\substack{0 \prec d_{r-1} \prec d_{r-2} \prec \ldots \prec d_1 \prec d_0 = 2^{r+1}-1 \\  \wt(d_k)=r+1-k,\ k=1,2,\ldots,r-1}} \sum_{0 \prec d_r \prec d_{r-1}}x^{d_{r-1}-d_r} y^{d_r}\prod_{i=1}^{r-1} a_i^{d_{i-1} - d_{i}}\Big)  \nonumber \\
& = & \tr_n\Big(\sum_{ \substack{0 \prec d_r \prec d_{r-1} \prec d_{r-2} \prec \ldots \prec d_1 \prec d_0 = 2^{r+1}-1 \\  \wt(d_k)=r+1-k,\ k=1,2,\ldots,r}} x^{d_{r-1}-d_r} y^{d_r}\prod_{i=1}^{r-1} a_i^{d_{i-1} - d_{i}}\Big). \label{equ:Bxy_di}
\end{eqnarray}

Let $e_i = d_{i-1} - d_i$ for $i = 1, 2, \ldots, r$. Let $e_{r+1} = d_r$. Note that 
\[
0 \prec d_r \prec d_{r-1} \prec d_{r-2} \prec \ldots \prec d_1 \prec d_0 = 2^{r+1}-1
\]
and $\wt(d_k)=r+1-k$ for $k=1,2,\ldots,r$. So $e_1, e_2, \ldots, e_{r+1}$ are distinct, and $\wt(e_k) = 1$ for $k = 1,2, \ldots, r+1$. Rewriting \eqref{equ:Bxy_di}, we have
\begin{equation}\label{equ:Bxy_di_update}
   B(x, y) = \tr_n\Big(\sum_{\substack{\text{distinct } e_{1},e_{2},\ldots,e_{r+1}\in \{2^0,2^1,\ldots,2^{r}\} \\ \wt(e_k)=1,\ \forall k\in \{ 1,2,\ldots,r+1\}}} (\prod_{i=1}^{r-1} a_i^{e_i}) x^{e_r}y^{e_{r+1}} \Big). 
\end{equation}

According to \eqref{equ:Bxy_di_update}, $B(x, y) = 0$ holds for all $y$ if and only if the coefficient of $y$ is zero, that is,

\begin{equation}
\label{equ:coe_y}
\sum_{\substack{\text{distinct }e_{1},e_{2},\ldots,e_{r+1}\in \{2^0,2^1,\ldots,2^{r}\} \\ \wt(e_k)=1,\ \forall k\in \{ 1,2,\ldots,r+1\}}} \left(a_1^{e_1}a_2^{e_2}a_3^{e_3}\ldots a_{r-1}^{e_{r-1}} x^{e_{r}}\right)^{e_{r+1}^{-1}}=0.
\end{equation}
Raising both sides of \eqref{equ:coe_y} to the $2^{r}$th power, we have 
\begin{equation}\label{multiple_term}
    \sum\limits_{\substack{\text{distinct } e_{1},e_{2},\ldots,e_{r+1} \in \{2^0,2^1,\ldots,2^{r}\} \\ \wt(e_k)=1,\ \forall k\in \{ 1,2,\ldots,r+1\}}} \left(a_1^{e_1}a_2^{e_2}a_3^{e_3}\ldots a_{r-1}^{e_{r-1}} x^{e_{r}}\right)^{2^r\cdot e_{r+1}^{-1}}=0.
\end{equation}
Observe that each monomial in the left hand side of \eqref{multiple_term} has degree at most $2^{2r}$, because $e_r \le 2^r$ and $2^r\cdot e_{r+1}^{-1} \le 2^r$. So the degree of \eqref{multiple_term} is at most $2^{2r}$, which implies that \eqref{multiple_term} has at most $2^{2r}$ solutions. Therefore, the dimension of the linear kernel of $B(x,y) $ is at most $2r$. By Lemma \ref{Walsh_spec_quar}, we have 

\begin{equation*}
    \nl(D_{a_{r-1}}D_{a_{r-2}}\ldots D_{a_{1}}f)\ge 2^{n-1}-2^{\frac{n+2r-2}{2}}. 
\end{equation*}
\end{proof}

We will need the following lemma in the proof of Theorem \ref{finv_nlr_lowerbound}.

\begin{lemma}\label{expansion_formula_multiple}
Let integer $r \ge 1$.
Let $\alpha_1 > \alpha_2 > \ldots > \alpha_r>0$ and $c_1,c_2,\ldots, c_r>0$. We have

\[
c_1\cdot 2^{\alpha_1 n}+c_2\cdot 2^{\alpha_2 n}+\ldots+c_r\cdot 2^{\alpha_r n} \le (\sqrt{c_1}\cdot 2^{\frac{1}{2}\cdot \alpha_1n}+ \frac{c_2}{2 \sqrt{c_1}}\cdot 2^{(\alpha_2-\frac{1}{2} \cdot \alpha_1)n} + \ldots +\frac{c_r}{2 \sqrt{c_1}}\cdot 2^{(\alpha_r-\frac{1}{2} \cdot \alpha_1)n} )^2
\]
\end{lemma}

\begin{proof}
By straightforward calculation, we have
\begin{eqnarray*}
\text{R.H.S} & = &  c_1\cdot 2^{\alpha_1 n} +\ldots+ c_r\cdot 2^{\alpha_r n} + \sum_{i=2}^{r}\frac{c_i^2}{4c_1}\cdot 2^{(2\alpha_i-\alpha_1)n}+ \sum_{i,j=2}^{r} \frac{c_i\cdot c_j}{2c_1}2^{(\alpha_i+\alpha_j-\alpha_1)n}\\
& \ge & \text{L.H.S}
\end{eqnarray*}

\end{proof}

In the following, we lower bound the $r$-th order nonlinearity for functions $\tr_n(x^{2^{r+1}-1})$.

\begin{theorem}\label{lower_nlr_high}
Let $f=\tr_n(x^{2^{r+1}-1})$ and $r \ge 1$. We have
\[
\nl_r(f)\ge 2^{n-1}-2^{(1-2^{-r})n+\frac{r}{2^{r-1}}-1}- O(2^{\frac{n}{2}}).
\]
\end{theorem}
\begin{proof} (of Theorem \ref{theorem_4})
Let $l_0=\nl_r(f)$ and 
\[
l_i=\min_{\text{distinct}\ a_1,\ldots,a_i\in \mathbb{F}_{2^n}^*} \nl_{r-i} (D_{a_i}\ldots D_{a_1}f)
\]
for $i = 1, 2, \ldots, r-1$.

By Proposition \ref{lower_non}, we have
\begin{eqnarray}
    l_{i}&=&\min_{\substack{ \text{distinct}\\ a_1,\ldots,a_{i}\in \mathbb{F}_{2^n}^*}}\nl_{r-i}(D_{a_{i}}\ldots D_{a_1}f) \nonumber\\
    &\ge& \min_{\substack{ \text{distinct}\\ a_1,\ldots,a_{i}\in \mathbb{F}_{2^n}^*}} 2^{n-1}-\frac{1}{2}\sqrt{2^{2n}-2\sum_{a_{i+1}\in \mathbb{F}_{2^n}^* \setminus\{a_1, a_2, \ldots, a_i\}}\nl_{r-i-1} (D_{a_{i+1}}\ldots D_{a_1}f)} \nonumber\\
    &\ge & 2^{n-1}-\frac{1}{2}\sqrt{2^{2n}-2(2^n-(i+1))l_{i+1}},
    \label{equ:l_i_ip1}
\end{eqnarray}
for $i = 0, 1, \ldots, r-2$.
Let $u_i=2^{n-1}-l_i$. Replacing $l_i$ by $2^{n-1}-u_i$ in \eqref{equ:l_i_ip1}, we have
\begin{equation}
\label{equ:u_i_ip1}
    u_i\le \frac{1}{2}\sqrt{2^n(i+1)+2^{n+1}u_{i+1}}.    
\end{equation}

\begin{claim}\label{l_i_def}
\begin{equation}
\label{equ:u_i_upper_bound}
u_{i}\le \frac{1}{2}\left(2^{(1 - 2^{-(r-i)})n+\frac{r}{2^{r-i-1}}}+ \sum_{j=1}^{r-i-1}(j+i)\cdot 2^{\frac{2^{j}-1}{2^{r-i}}n-\frac{2^{j}-1}{2^{r-i-1}}r-j} \right).
\end{equation}
for $0\le i\le r-2$.
\end{claim}
\begin{proof} (of Claim \ref{l_i_def})
We prove by induction on $i$. For the base step, we prove the claim for $i=r-2$. By \eqref{equ:u_i_ip1}, we have 
\begin{eqnarray}\label{eq:u_r_2}
    u_{r-2}\le \frac{1}{2}\sqrt{2^n(r-1)+2^{n+1}u_{r-1}} .
\end{eqnarray}

By definition of $l_{r-1}$ and Lemma \ref{nonliearity_derivative}, we have
$l_{r-1}\ge 2^{n-1}-2^{\frac{n+2r-2}{2}}$, that is, $u_{r-1} \le  2^{\frac{n+2r-2}{2}}$. Plugging $u_{r-1} \le  2^{\frac{n+2r-2}{2}}$ into \eqref{eq:u_r_2}, we have
\begin{eqnarray*}
   u_{r-2}&\le &
   \frac{1}{2}\sqrt{2^{\frac{1}{2}(3n+2r)}+(r-1)2^n}
   \\ &\le & \frac{1}{2}(2^{\frac{3}{4}n+\frac{r}{2}}+(r-1)2^{\frac{n}{4}-\frac{r}{2}-1}),
\end{eqnarray*}
where the last step follows from Lemma \ref{expansion_formula_multiple}.

For the induction step, assuming inequality \eqref{equ:u_i_upper_bound} holds for $i+1$, we prove \eqref{equ:u_i_upper_bound} for $i$, where $i = r-3, r-4, \ldots, 0$. Assuming \eqref{equ:u_i_upper_bound} is true for $i+1$, we prove it for $i$. We have
 \begin{eqnarray*}
    u_i&\le& \frac{1}{2}\sqrt{2^n(i+1)+2^{n+1}u_{i+1}}\\
    &\le & \frac{1}{2}\sqrt{2^n(i+1)+ 2^{n}\cdot\left(2^{(1 - 2^{-(r-i-1)})n+\frac{r}{2^{r-i-2}}}+ \sum_{j=1}^{r-i-2}(j+i+1)\cdot 2^{\frac{2^{j}-1}{2^{r-i-1}}n-\frac{2^{j}-1}{2^{r-i-2}}r-j} \right)} \\
    &\le & \frac{1}{2}\left( 2^{(1-2^{-(r-i)})n+\frac{r}{2^{r-i-1}}}+ \sum_{j=1}^{r-i-1}(j+i)\cdot 2^{\frac{2^{j}-1}{2^{r-i}}n-\frac{2^{j}-1}{2^{r-i-1}}r-j}\right),
\end{eqnarray*}
as desired, where the third step follows from Lemma \ref{expansion_formula_multiple}.
\end{proof}

Turn back to the proof of Theorem \ref{theorem_4}. By Claim \ref{l_i_def}, we have 
\begin{eqnarray*}\label{nl_r_result}
    \nl_r(f)&=&2^{n-1}-u_0 \\
    &\ge & 2^{n-1}-2^{(1-2^{-r})n+\frac{r}{2^{r-1}}-1}- \sum_{j=1}^{r-1}j\cdot 2^{\frac{2^{j}-1}{2^{r}}n-\frac{2^{j}-1}{2^{r-1}}r-(j+1)}\\
    &\ge & 2^{n-1}-2^{(1-2^{-r})n+\frac{r}{2^{r-1}}-1}-O(2^{\frac{n}{2}}).
\end{eqnarray*}
\end{proof}

\begin{remark}
    By Theorem \ref{lower_nlr_high}, we deduce that
    \begin{eqnarray*}
            \nl_r(f)&\ge& 2^{n-1}-2^{(1-2^{-r})n+\frac{r}{2^{r-1}}-1}- O(2^{\frac{n}{2}})
            \\& =& 2^{n-1}(1-\exp(-\frac{\alpha\cdot n}{2^r})),
    \end{eqnarray*}
where $\alpha\approx \log_2 e$ when $ r\ll \log_2 n$.
\end{remark}

Similarly, for the inverse function, we prove the following nonlinearity lower bound. This is studied by Carlet in \cite{Car08}, who claims that the $r$-th order nonlinearity is asymptotically lower bounded by $2^{n-1}-2^{(1-2^{-r})n} $. We credit the lower bound, i.e., Theorem \ref{finv_nlr_lowerbound}, to Carlet, since our proof closely follows the method in \cite{Car08} by working out the calculations carefully. The proof of the following theorem is in Appendix \ref{high_non_inverse}. 

\begin{theorem}\label{finv_nlr_lowerbound}
Let $f_{\mathrm{inv}}=\tr_n(x^{2^n-2})$. For any $r \ge 1$, we have
    $ \nl_r(f_{\mathrm{inv}})  \ge  2^{n-1}-2^{(1-2^{-r})n-2^{-(r-1)}}- O(2^{\frac{n}{2}})$.
\end{theorem}

Note that the bound in Theorem \ref{lower_nlr_high} is slightly better than that in Theorem \ref{finv_nlr_lowerbound}.

\subsection{Comparison}
 Babai, Nisan and Szegedy \cite{BNS92} proved that  the $r$-th nonlinearity of the generalized inner product function
 \[
 \mathrm{GIP}_{r+1}(x_1, x_2, \ldots, x_n) = \prod_{i=1}^{r+1}x_i + \prod_{i=r+2}^{2(r+1)}x_i + \ldots + \prod_{i=n-r}^{n}x_i
 \]
 is lower bounded by $2^{n-1}(1-\mathrm{exp}(-\Omega(\frac{n}{r\cdot 4^r})))$. Bourgain \cite{Bou05} and Green \emph{et al.} \cite{GRS05} proved that the $r$-th nonlinearity of the $\mathrm{mod}_3$ function is at least $2^{n-1}(1-\mathrm{exp}(-\frac{n}{8^r}))$; Viola \cite{Vio06} and Chattopadhyay \cite{Cha07} improved this bound to $2^{n-1}(1-\mathrm{exp}(-\frac{n}{4^r}))$. Viola \cite{Vio06} exhibited an explicit function $f \in P$ (which relies on explicit small-bias generators) with $r$-th nonlinearity at least $2^{n-1}(1-\mathrm{exp}(-\frac{\alpha \cdot n}{2^r}))$, where $\alpha < \frac{1}{4}\cdot \log_2 e$; the lower bound is also proved in \cite{VW08} using similar argument. 

By Theorem \ref{lower_nlr_high}, we prove that the $r$-th order nonlinearity of $\tr_n(x^{2^{r+1}-1})$ is at least $2^{n-1}(1-\exp(-\frac{\beta\cdot n}{2^r}))$, where $\beta\approx \log_2 e$ when $r\ll \log_2 n$. Previous to our work, the best lower bound is $2^{n-1}(1-\mathrm{exp}(-\frac{\alpha \cdot n}{2^r}))$ \cite{Vio06, VW08}, where $\alpha < \frac{1}{4}\cdot \log_2 e$.

\section{Conclusion}
Using algebraic methods, we lower bound the second-order, third-order, and higher-order nonlinearities of some trace monomial Boolean functions. For the second-order nonlinearity, we study Boolean functions $\tr_n(x^7)$ and $\tr_n(x^{2^r+3})$ for $n=2r$; the latter class of Boolean functions is studied for the first time. Our lower bounds match the best proven lower bounds on the second-order nonlinearity among all trace monomial functions \cite{Car08, YT20}. For the third-order nonlinearity, we prove the lower bound for functions $\tr_n(x^{15})$, which is the best provable third-order nonlinearity lower bound. For higher-order nonlinearity, we prove the lower bound
\[
\nl_r(f)\ge 2^{n-1}-2^{(1-2^{-r})n+\frac{r}{2^{r-1}}-1}-O(2^{\frac{n}{2}})
\]
for functions $\tr_n(x^{2^{r+1}-1})$. When $r \ll \log n$, this is the best lower bound, compared with all the previous works, e.g., \cite{BNS92, Bou05, GRS05, Cha07, Vio06, VW08}.

\printbibliography

\appendix

\section{Proof of Lemma \ref{number_roots}}
\label{app:proof_of_lemma_number_roots}
Let $c=a^{-5}$. Let $h(b)=\sum\limits_{\substack{d_1\in \{0,1\} \\ \cdots \\ d_{r-1}\in \{0,1\}}} \sum_{j=0}^{n-1} c^{2^j} b^{\sum_{i=1}^{r-1}2^{n-2i+d_i+j}}$. We know $\{b^i \mid 1 \le i \le 2^n-1, b\in \mathbb{F}_{2^n}^*\}$ is a multiplicative group of order $ 2^n-1$. We prove there exists only one monomial of $h(b)$ with the maximum degree $ \frac{5}{3} \cdot 2^{n-1} - \frac{32}{3}$. Let $d_1=0$, $d_i=1 $, $j=1$ for $i\ge 2$. We have a monomial $c^2b^{2^{n-1}+\sum_{i=2}^{r-1}2^{n-2i+2}} =c^2b^{2^{n-1}+2^{n-2}+2^{n-4}+\ldots +2^5}=c^2b^{\frac{5}{3} \cdot 2^{n-1} - \frac{32}{3}} $. Next, we prove there are no other monomials with degree $\frac{5}{3} \cdot 2^{n-1} - \frac{32}{3}$.

\textbf{Case 1}: $j=0$. In this case, we deduce that the degree of the corresponding monomials $cb^{\sum_{i=1}^{r-1} 2^{n-2i+d_i}}$ for $d_i\in \{0,1\}$ and $1\le i\le r-1$ is 
    \begin{eqnarray*}
        &&2^{n-2+d_1}+2^{n-4+d_2}+\ldots+2^{3+d_{r-1}}\\ &\le & 2^{n-1}+2^{n-3}+\ldots+ 2^4
        \\ &= & \frac{2^{n+1}-16}{3}.    
     \end{eqnarray*}
When $n\ge 6$, we have $\frac{2^{n+1}-16}{3} < \frac{5}{3} \cdot 2^{n-1} - \frac{32}{3}$.

\textbf{Case 2}: $d_1=0$ and  $j=1$. In this case, there is a monomial $c^2b^{2^{n-1}+\sum_{i=2}^{r-1} 2^{n-2i+1+d_i}}$ , the degree of which is
    \begin{eqnarray*}
            & & 2^{n-1}+\sum_{i=2}^{r-1} 2^{n-2i+1+d_i}
            \\& \le & 2^{n-1}+\sum_{i=2}^{r-1} 2^{n-2i+2} 
            \\& = & \frac{5}{3} \cdot 2^{n-1} - \frac{32}{3}.   
    \end{eqnarray*}
The equality holds if and only if $d_1=0$, $d_i=1 $, $j=1$ for $i\ge 2$.

\textbf{Case 3}: $d_1+j\ge 2$ and $3+d_{r-1}+j\le n-1$. In this case, we have $\sum_{i=1}^{r-1} 2^{n-2i+d_i+j}\ge 2^n$ for some $i$ and $j$. Let $i_0$ be such that $n-2i_0+d_{i_0}+j\ge n$ and $n-2(i_0+1)+d_{i_0+1}+j\le n-1 $. We can deduce that $n-2\le n-2i_0-2+d_{i_0+1}+j\le n-1 $. 
Since $b^{2^n}=b$, we have 
\begin{eqnarray}\label{case3}
        c^{ 2^j}b^{\sum_{i=1}^{r-1}2^{n-2i+d_i+j} }&= 
        & c^{ 2^j}(b^{\sum_{i=1}^{i_0} 2^{n-2i+d_i+j}}\cdot b^{\sum_{i={i_0}+1}^{r-1} 2^{n-2i+d_i+j}}) \nonumber
        \\&=&c^{ 2^j} (b^{\sum_{i=1}^{i_0} 2^{-2i+d_i+j}} \cdot b^{\sum_{i={i_0}+1}^{r-1} 2^{n-2i+d_i+j} } ).
\end{eqnarray}
    Let $d= \sum_{i=1}^{i_0} 2^{-2i+d_i+j} + \sum_{i=i_0+1}^{r-1} 2^{n-2i+d_i+j}$ denote by the degree of the polynomial \eqref{case3}.
    For a fixed $j$, we have $2^{n-2i+d_i+j}$ is decreasing for $i_0+1\le i\le r-1$ and $d_i\in \{0,1\} $ and $2^{-2i+d_i+j}$ is also decreasing for $1\le i\le i_0$ and $d_i\in \{0,1\} $. That is,
    \begin{equation}\label{recurise_eq}
        2^{n-1}\ge 2^{n-2i_0-2+d_{i_0+1}+j}>\ldots>2^{3+d_{r-1}+j}>2^{d_1+j-2}>\ldots>2^{d_{i_0}+j-2i_0}.
    \end{equation}

\indent \textbf{Subcase 3.1}: $n-2i_0-2+d_{i_0+1}+j=n-2$. We have 
\begin{eqnarray*}
         d&=&  \sum_{i=i_0+1}^{r-1} 2^{n-2i+d_i+j}+\sum_{i=1}^{i_0} 2^{-2i+d_i+j}
         \\& < &2^{n-1}<\frac{5}{3} \cdot 2^{n-1} - \frac{32}{3}
\end{eqnarray*}
 for $n\ge 6 $.

\indent \textbf{Subcase 3.2}:  $n-2i_0-2+d_{i_0+1}+j=n-1$ and $i_0=r-2$. We have
    \begin{equation*}
    \begin{aligned}
            d\quad &= \quad  2^{n-1}+ \sum_{i=1}^{r-2} 2^{-2i+d_i+j} 
            \\ \quad &< \quad 2^{n-1}+2^{d_1+j-1} 
            \\ \quad & \le \quad 2^{n-1}+2^{n-5-d_{r-1}+d_1} && \qquad \text{By $3+d_{r-1}+j\le n-1$}
            \\ \quad& \le \quad 2^{n-1}+2^{n-4} 
    \end{aligned}
     \end{equation*}

For $n\ge 6$, we have $d\le 2^{n-1}+2^{n-4}< \frac{5}{3} \cdot 2^{n-1} - \frac{32}{3}$.

\textbf{Subcase 3.3}: $n-2i_0-2+d_{i_0+1}+j=n-1$ and $i_0<r-2$. We have
    \begin{equation*}
    \begin{aligned}
        d \quad &= \quad \sum_{i=1}^{r-1-i_0} 2^{n-2i_0-2i+d_{i_0+i}+j} +\sum_{i=1}^{i_0}2^{-2i+d_i+j}
            \\&= \quad \sum_{i=1}^{r-1-i_0} 2^{n-2i+1+d_{i_0+i}-d_{i_0+1}} +\sum_{i=1}^{i_0}2^{-2i+d_i+j} 
            \\ &\le \quad 2^{n-1}+(\sum_{i=2}^{r-1-i_0}2^{n-2i+2})+2^{d_1+j-1} && \qquad \text{By $ \sum_{i=1}^{i_0}2^{-2i+d_i+j} < 2^{d_1+j-1}$ }
            \\&= \quad 2^{n-1} +( \sum_{i=2}^{r-1-i_0}2^{n-2i+2})+ 2^{2i_0+d_1-d_{i_0+1}} && \qquad \text{By  $n-2i_0-2+d_{i_0+1}+j=n-1$ }
            \\ &\le \quad 2^{n-1}+( \sum_{i=2}^{r-1-i_0}2^{n-2i+2})+2^{2i_0+1} && \qquad \text{By  $d_1-d_{i_0+1}\le 1 $ } 
            \\&< \quad 2^{n-1}+\sum_{i=2}^{r-i_0} 2^{n-2i+2}.
    \end{aligned}
    \end{equation*}

Since $ 2^{n-1}+\sum_{i=2}^{r-i_0} 2^{n-2i+2} \le 2^{n-1}+\sum_{i=2}^{r-1} 2^{n-2i+2} = \frac{5}{3} \cdot 2^{n-1} - \frac{32}{3}$, we have $d< \frac{5}{3} \cdot 2^{n-1} - \frac{32}{3}$.

\textbf{Case 4}: $3+d_{r-1}+j\ge n$. In this case, all the exponents $2^{n-2i+d_i+j}$ of the corresponding monomials' degree of $h(b)$ are greater than $2^n$. So the degree of the corresponding monomials can be represented as 
 \begin{equation*}
 \begin{aligned}
    &  \sum_{i=1}^{r-1} 2^{-2i+d_i+j}
         \\ \le \quad & \sum_{i=1}^{r-1} 2^{n-2i+d_i-1} && \qquad \text{By $j\le n-1$}
         \\ \le \quad & \sum_{i=1}^{r-1} 2^{n-2i}  &&\qquad \text{Since $d_i\le 1$}
         \\ < \quad & 2^{n-1}
         \\< \quad & \frac{5}{3} \cdot 2^{n-1} - \frac{32}{3}
 \end{aligned}
\end{equation*}
for $n\ge 6$.

Summarizing all cases, we know that the degree of the polynomial $h(b)$ is exactly $\frac{5}{3} \cdot 2^{n-1} - \frac{32}{3}$. 

\section{ Proof of Lemma \ref{minimum_degree}}
\label{app:proof_of_lem_min_degree}

Let $c=a^{-5}$. Let $h(b)=\sum\limits_{\substack{d_1\in \{0,1\} \\ \cdots \\ d_{r-1}\in \{0,1\}}} \sum_{j=0}^{n-1} c^{ 2^j} b^{\sum_{i=1}^{r-1}2^{n-2i+d_i+j}}$. Let $d_i=0,d_{r-1}=1,j=n-4$ for $i< r-1$, we have a monomial $c^{ 2^{n-4}}b^{\sum_{i=1}^{r-2}2^{2n-2i-4}+2^n}=c^{2^{n-4}}b^{\sum_{i=1}^{r-2}2^{n-2i-4}+1}=c^{2^{n-4}}b^{\frac{2^{n-4}+1}{3}}$. Next, we prove that there are no other monomials with the minimum degree $\le \frac{2^{n-4}+1}{3}$.
    
    \textbf{Case 1}: $3+d_{r-1}+j<n$. In this case, we have $2^{n-2+d_1+j}\ge2^{n-2}$, $2^{3+d_{r-1}+j}\le 2^{n-1}$. Let $i_0$ be such that $n-2i_0+d_{i_0}+j\ge n$ and $n-2(i_0+1)+d_{i_0+1}+j= n-1 $ or $n-2$. Then we have
    \begin{eqnarray*}
        c^{2^j}b^{\sum_{i=1}^{r-1}2^{n-2i+d_i+j} }&= 
        & c^{ 2^j}(b^{\sum_{i=1}^{i_0} 2^{n-2i+d_i+j}}\cdot b^{\sum_{i={i_0}+1}^{r-1} 2^{n-2i+d_i+j}})
        \\&=&c^{ 2^j} (b^{\sum_{i=1}^{i_0} 2^{-2i+d_i+j}} \cdot b^{\sum_{i={i_0}+1}^{r-1} 2^{n-2i+d_i+j} } ).
    \end{eqnarray*}
    
    According to \eqref{recurise_eq}, we deduce that the degree of the above monomial is 
        \begin{eqnarray*}
             2^n> \sum_{i=1}^{r-1}2^{n-2i+d_i+j} &=& \sum_{i=1}^{i_0} 2^{-2i+d_i+j}+\sum_{i={i_0}+1}^{r-1} 2^{n-2i+d_i+j}
            \\ &>   &2^{n-2} 
            >  \frac{2^{n-4}+1}{3}.
        \end{eqnarray*}

\textbf{Case 2}: $3+d_{r-1}+j\ge n$. In this case, the degree of the following monomials is
        \begin{equation*}
        \begin{aligned}
            \quad  \sum_{i=1}^{r-1}2^{n-2i+d_i+j}   =  & \quad \sum_{i=1}^{r-1}2^{-2i+d_i+j} \\  \ge &  \quad (\sum_{i=1}^{r-2}2^{n-2i-3+d_i-d_{r-1}})+2^{3+d_{r-1}+j-n} && \qquad \text{By $3+d_{r-1}+j\ge n $} \\ 
             \ge & \quad \sum_{i=1}^{r-2}2^{n-2i-4}+2^{0} \\
             = & \quad \frac{2^{n-4}+1}{3}.
        \end{aligned}
        \end{equation*}
The equality holds if and only if $d_i=0$, $d_{r-1}=1$ and $j=n-4$ where $0 \le i \le r-2$.

Summarizing all cases, we prove that the minimum degree of the monomial of $h(b)$ is $\frac{2^{n-4}+1}{3}$.

\section{Proof of Theorem \ref{finv_nlr_lowerbound}} \label{high_non_inverse}
Carlet \cite{Car08} gave the expression of the lower bound on the $r$-th order nonlinearity of the inverse function $\tr_n(x^{2^n-2})$ by applying iteratively the lower bound on the third-order nonlinearity. 
\begin{proposition} \label{nl3_lower_bound} \cite[Proposition 7]{Car08}
    Let $f_{\mathrm{inv}}(x)=\tr_n(x^{2^n-2})$. Then we have
    \footnote{There is a mistake of the lower bound in \cite{Car08}, which has been corrected here.}   
    \begin{eqnarray*}
        \nl_3(f_{\mathrm{inv}})&\ge& 2^{n-1}-\frac{1}{2}\sqrt{ (2^n-1)\sqrt{2^{\frac{3n}{2}+3}+3\cdot 2^{n+1}-2^{\frac{n}{2}+3}+16 }+2^n}.
    \end{eqnarray*} 
\end{proposition}

\begin{proposition} \label{nlr_lower_bound} \cite[page 1271]{Car08}
    Let $f_{\mathrm{inv}}(x)=\tr_n(x^{2^n-2})$ and $r\ge 1$. Then 
    \begin{equation}\label{finv_recursion}
             \nl_r(f_{\mathrm{inv}})\ge 2^{n-1}-l_r,
    \end{equation} 
where $l_r=\sqrt{(2^n-1)(l_{r-1}+1)+2^{n-2} }$.
\end{proposition}

\begin{lemma}\label{l_r_bound}
For any $3\le r\le n-3$, we have
\begin{eqnarray}\label{eq_lr}
     l_r \le 2^{(1-2^{-r})n-2^{-(r-1)}}+3\cdot 2^{(\frac{1}{2}-2^{-r})n}.
 \end{eqnarray} 
\end{lemma}
\begin{proof}
We prove by induction on $r$. For the base case, we prove \eqref{eq_lr} for $r=3$ by \eqref{finv_l_3}.
By Proposition \ref{nl3_lower_bound}, we have 
\begin{eqnarray}
    l_3&=& \frac{1}{2}\sqrt{(2^n-1)\sqrt{2^{\frac{3n}{2}+3}+3\cdot 2^{n+1}-2^{\frac{n}{2}+3}+16}+2^n}\nonumber \\
    & \le &  \frac{1}{2}\sqrt{2^n\cdot (2^{\frac{3n}{4}+\frac{3}{2}}+3\cdot 2^{-\frac{3}{2}}\cdot 2^{\frac{n}{4}})+2^n}\nonumber\\
    &\le & \frac{1}{2}(2^{\frac{7n}{8}+\frac{3}{4}}+3\cdot 2^{\frac{3n}{8}-\frac{13}{4}}+2^{\frac{n}{8}-\frac{7}{4}})\nonumber \\
    &\le & 2^{\frac{7n}{8}-\frac{1}{4}}+3\cdot 2^{\frac{3n}{8}} \label{finv_l_3},
\end{eqnarray}
as desired, where the second and third steps follow from Lemma \ref{expansion_formula_multiple}.

For the induction step, assuming inequality \eqref{eq_lr} holds for $r$, we prove it for $r+1$, where $r = 3, 4, \ldots, n-4$. We have
\begin{eqnarray*}
    l_{r+1}&=& \sqrt{(2^n-1)(l_r+1)+2^{n-2}}
        \\ &\le & \sqrt{2^n (l_r+1)+2^{n-2}}\\
        & \le & \sqrt{2^n(2^{(1-2^{-r})n-2^{-(r-1)}}+3\cdot 2^{(\frac{1}{2}-2^{-r})n} +1 )+2^{n-2} } \\
        &\le & 2^{(1-2^{-(r+1)})n-2^{-r}}+3\cdot 2^{(\frac{1}{2}-2^{-(r+1)})n-1+2^{-r}} + 5\cdot 2^{2^{-(r+1)}n-3+2^{-r}}\\
        &\le & 2^{(1-2^{-(r+1)})n-2^{-r}}+3\cdot 2^{(\frac{1}{2}-2^{-(r+1)})n},
\end{eqnarray*}
where the penultimate step is by Lemma \ref{expansion_formula_multiple}. We have completed the induction step, and thus the proof.
\end{proof}

Now we are ready to prove the lower bound on the $r$-th order nonlinearity of functions $\tr_n(x^{2^n-2})$. 

\begin{proof} (of Theorem \ref{thm:inverse_lb})
By Proposition \ref{nlr_lower_bound} and Lemma \ref{l_r_bound}, we have 
\begin{eqnarray*} 
        \nl_r(f_{\mathrm{inv}})  &= & 2^{n-1}-l_r
        \\ &\ge & 2^{n-1}- 2^{(1-2^{-r})n-2^{-(r-1)}}- 3\cdot 2^{(\frac{1}{2}-2^{-r})n}
        \\ &\ge & 2^{n-1}-2^{(1-2^{-r})n-2^{-(r-1)}}- O(2^{\frac{n}{2}}).
\end{eqnarray*}
\end{proof}

\end{document}